\documentclass[11pt]{amsart}
\usepackage[T1]{fontenc}
\usepackage{amssymb,amsmath, amsthm, amsfonts, tikz-cd}

\usepackage{graphicx}
\usepackage{listings}
\usepackage[margin=.75in]{geometry}
\usepackage{lstautogobble}
\usepackage{enumerate}
\usepackage[shortlabels]{enumitem}
\usepackage{thmtools}
\usepackage{thm-restate}
\usepackage{amsthm}
\usepackage{verbatim}

\usepackage{mathtools}
\usepackage{physics}
\usepackage{color}
\usepackage{hyperref}
\usepackage[capitalise]{cleveref}
\usepackage[final]{showlabels}
\usepackage[bottom]{footmisc}
\crefformat{equation}{(#2#1#3)}

\usepackage{float}
\restylefloat{table}

\usepackage{mathrsfs}
\setlist{  
  listparindent=\parindent,
  parsep=0pt,
}

\theoremstyle{plain}
\newtheorem{thm}{Theorem}[section]
\newtheorem{prop}[thm]{Proposition}
\newtheorem{lemma}[thm]{Lemma}

\theoremstyle{definition}
\newtheorem{mydef}[thm]{Definition}

\newtheorem{remark}[thm]{Remark}

\numberwithin{equation}{section} 

\DeclarePairedDelimiter\ipp{\langle}{\rangle}

\DeclarePairedDelimiter{\paren}{\lparen}{\rparen}

\DeclarePairedDelimiter{\jp}{\langle}{\rangle}

\DeclareMathOperator{\sgn}{sgn}

\newcommand{\p}{{\partial}}

\renewcommand{\d}{\delta}

\newcommand{\R}{{\mathbb{R}}}
\newcommand{\C}{{\mathbb{C}}}
\newcommand{\N}{{\mathbb{N}}}

\newcommand{\Z}{{\mathbb{Z}}}

\newcommand{\Sc}{{\mathcal{S}}}

\newcommand{\tl}{\tilde}

\newcommand{\D}{\Delta}

\newcommand{\ph}{\phantom{=}}
\newcommand{\nn}{\nonumber}

\newcommand{\ol}{\overline}
\newcommand{\ul}{\underline}

\newcommand{\ux}{\underline{x}}

\newcommand{\ua}{\underline{\alpha}}

\newcommand{\ep}{\epsilon}
\newcommand{\vep}{\varepsilon}
\newcommand{\al}{\alpha}

\newcommand{\id}{\mathbf{1}}

\newcommand{\wh}{\widehat}

\newcommand{\uD}{\underline{\Delta}}
\newcommand{\uxi}{\underline{\xi}}

\newcommand{\F}{{\mathcal{F}}}

\let\oldtocsection=\tocsection
 
\let\oldtocsubsection=\tocsubsection
 
\let\oldtocsubsubsection=\tocsubsubsection
 
\renewcommand{\tocsection}[2]{\hspace{0em}\oldtocsection{#1}{#2}}
\renewcommand{\tocsubsection}[2]{\hspace{1em}\oldtocsubsection{#1}{#2}}
\renewcommand{\tocsubsubsection}[2]{\hspace{2em}\oldtocsubsubsection{#1}{#2}}

\title[The Mean-Field Limit of the Lieb-Liniger Model]{The Mean-Field Limit of the Lieb-Liniger Model}

\author[M. Rosenzweig]{Matthew Rosenzweig}
\address{  
Massachusetts Institute of Technology\\ 
Department of Mathematics\\ 
Headquarters Office\\
Simons Building (Building 2), Room 106\\
77 Massachusetts Avenue\\
Cambridge, MA 02139-4307}
\email{mrosenzw@mit.edu}

\begin{document}
\begin{abstract}
We consider the well-known Lieb-Liniger (LL) model for $N$ bosons interacting pairwise on the line via the $\delta$-potential in the mean-field scaling regime. Assuming suitable asymptotic factorization of the initial wave functions and convergence of the microscopic energy per particle, we show that the time-dependent reduced density matrices of the system converge in trace norm to the pure states given by the solution to the one-dimensional cubic nonlinear Schr\"odinger equation (NLS) with an  explict rate of convergence. In contrast to previous work \cite{AmBre2012} relying on quantum field theory and without an explicit rate, our proof is inspired by the counting method of Pickl \cite{Pickl2010, Pickl2011, Pickl2015} and Knowles and Pickl \cite{KP2010}. To overcome difficulties stemming from the singularity of the $\delta$-potential, we introduce a new short-range approximation argument that exploits the H\"older continuity of the $N$-body wave function in a single particle variable. By further exploiting the $L^2$-subcritical well-posedness theory for the 1D cubic NLS, we can prove mean-field convergence assuming only that the limiting solution to the NLS has finite mass.
\end{abstract}

\maketitle{}

\section{Introduction}
\label{sec:intro}
\subsection{Background}
The \emph{Lieb-Liniger (LL) model}  describes a finite number of bosons in one dimension with two-body contact interactions. Formally, the Hamiltonian for $N$ bosons is given by
\begin{equation}
\label{eq:LL_ham_orig}
\sum_{i=1}^N  -\Delta_{i} + c\sum_{1\leq i<j\leq N}\delta(X_i-X_j),
\end{equation}
where $-\Delta_i$ denotes the Laplacian in the $i$-th particle variable $x_i\in\R$, $\delta(X_i-X_j)$ denotes multiplication by the distribution $\delta(x_i-x_j)$, and $c\in\R$ is the coupling constant determining the strength of the interaction and whether it is repulsive ($c>0$) or attractive ($c<0$). The LL model is named for Lieb and Liniger, who showed in the seminal works \cite{LL1963_I, LL1963_II} that when considered on a finite interval with periodic boundary conditions, the model is exactly solvable by Bethe ansatz.\footnote{Bethe ansatz refers to a method in the study of exactly solvable models originally introduced by Hans Bethe to find exact eigenvalues and eigenvectors of the antiferromagnetic Heisenberg spin chain \cite{Bethe1931}. For more on this technique and its applications, we refer the reader to the monograph \cite{Gaudin2014}.} While it was originally introduced as a toy quantum many-body system, the LL model has since attracted interest from both the physics community \cite{Olshanii1998, PSW2000, DLO2001,JK2002,LSY2003, OD2003, PGS2004} and the mathematics community \cite{LSY2004, SY2008} in modeling quasi-one-dimensional dilute Bose gases which have been realized in laboratory settings \cite{DHR2001, RGTHBA2003, TOHPRP2004, ETSAWB2006}.

In applications, the number of particles $N$ is large, ranging upwards from $N\approx 10^3$ in the case of very dilute Bose-Einstein condensates. For large $N$, it is computationally expensive to extract useful information about the time evolution of the system directly from its wave function. Thus, one seeks to find an evolution equation, for which one can more efficiently extract information, that provides an effective description of the $N$-body system for large values of $N$. To obtain nontrivial dynamics in the limit as $N\rightarrow\infty$, we consider the \emph{mean-field} scaling regime, where the coupling constant $c$ in \eqref{eq:LL_ham_orig} is taken to be equal to $\kappa/N$ for $\kappa\in\{\pm 1\}$, so that the Hamiltonian becomes
\begin{equation}
\label{eq:LL_ham}
H_N \coloneqq \sum_{i=1}^N -\D_i + \frac{\kappa}{N}\sum_{1\leq i<j\leq N}\d(X_i-X_j), \qquad \kappa\in\{\pm 1\}.
\end{equation}
Note that the mean-field scaling is such that the free and interacting components of the Hamiltonian $H_{N}$ are of the same order in $N$. By means of quadratic forms (e.g. \cite[Chapter X]{RSII}, \cite[Section 3]{AmBre2012}), the expression \eqref{eq:LL_ham} can be realized as a self-adjoint operator on the Hilbert space $L_{sym}^2(\R^N)$ consisting of wave functions symmetric under permutation of particle labels. By Stone's theorem, the corresponding Schr\"odinger problem
\begin{equation}
\label{eq:N_Schr}
\begin{cases}
i\p_t\Phi_N = H_{N}\Phi_N\\
\Phi_N|_{t=0} = \Phi_{N,0}
\end{cases}
\end{equation}
has a unique global solution $\Phi_N(t) = e^{-it H_{N}}\Phi_{N,0}$. Of particular interest are \emph{factorized} initial data $\Phi_{N,0} =\phi_0^{\otimes N}$, for $\phi_0\in L^2(\R)$ satisfying $\|\phi_0\|_{L^2(\R)}=1$, which correspond to a system where the $N$ particles are all in the same initial state $\phi_0$.

In general, factorization of the wave function $\Phi_N$ is not preserved by the time evolution due to the interaction between particles. However, it is reasonable to expect from the scaling in \eqref{eq:LL_ham} that the total potential experienced by each particle is approximately described by an effective \emph{mean-field} potential in the limit as $N\rightarrow\infty$. Formally, we may expect that
\begin{equation}
\label{eq:wf_approx}
\Phi_N \approx \phi^{\otimes N} \quad \text{as $N\rightarrow\infty$},
\end{equation}
for some $\phi:\R\times\R\rightarrow\C$, in some sense to be made precise momentarily. To find an equation satisfied by $\phi$ and to give rigorous meaning to the approximation \eqref{eq:wf_approx}, we argue as follows. Let $\Phi_N$ be the solution to the Schr\"odinger equation \eqref{eq:N_Schr}, and consider the \emph{density matrix}
\begin{equation}
\Psi_N \coloneqq \ket*{\Phi_N}\bra*{\Phi_N}
\end{equation}
associated to $\Phi_N$.\footnote{Here and in the sequel, we use Dirac's bra-ket notation: for $f,g,h\in L^2(\R^d)$, the operator $\ket*{f}\bra*{g}:L^2(\R^d)\rightarrow L^2(\R^d)$ is defined by $(\ket*{f}\bra*{g})h = \ip{g}{h}_{L^2}f$. The integral kernel of $\ket*{f}\bra*{g}$ is $f(x)\ol{g(x')}$.\label{fn:Dirac}} We define the \emph{$k$-particle reduced density matrix} $\gamma_N^{(k)}$ associated to $\Phi_N$ by
\begin{equation}
\label{eq:RDM}
\gamma_N^{(k)} \coloneqq \Tr_{k+1,\ldots,N}\Psi_N \qquad k\in\{1,\ldots,N\},
\end{equation}
where $\Tr_{k+1,\ldots,N}$ denotes the partial trace over the coordinates $(x_{k+1},\ldots,x_N)$. Using equation \eqref{eq:N_Schr}, one can show that $\{\gamma_N^{(k)}\}_{k=1}^N$ formally converges, as $N\rightarrow\infty$, to a solution $\{\gamma_k\}_{k=1}^\infty$ of the \emph{Gross-Pitaevskii (GP) hierarchy}:
\begin{equation}
\label{eq:GP}
i\p_t\gamma^{(k)} = \comm{-\uD_k}{\gamma^{(k)}} + \kappa\sum_{j=1}^k \Tr_{k+1}\paren*{\comm{\delta(X_j-X_{k+1})}{\gamma^{(k+1)}}},
\end{equation}
where $\uD_k \coloneqq \sum_{i=1}^k \D_i$ and $\comm{\cdot}{\cdot}$ denotes the usual commutator bracket. It is a short computation that if the GP solution takes the form $\gamma^{(k)}=\ket*{\phi^{\otimes k}}\bra*{\phi^{\otimes k}}$ for every $k\in\N$, then the one-particle wave function $\phi$ solves the \emph{one-dimensional (1D) cubic nonlinear Schr\"odinger (NLS) equation}
\begin{equation}
\label{eq:NLS}
\begin{cases}
(i\p_t+\Delta)\phi = \kappa|\phi|^2\phi \\
\phi|_{t=0} =\phi_0
\end{cases}, \qquad \kappa\in\{\pm 1\}.
\end{equation}
Thus, we formally refer to the 1D cubic NLS as the \emph{mean-field limit} of the LL model.\footnote{Just as the LL model is exactly solvable by Bethe ansatz, the 1D cubic NLS is exactly solvable by the inverse scattering transform \cite{ZS72, FT07}; a relationship we discuss in our joint work \cite{MNPRS2_2019}.} To rigorously establish the validity of the mean-field approximation, one needs to show convergence of the $k$-particle reduced density matrices $\gamma_N^{(k)}$ to $\ket*{\phi^{\otimes k}}\bra*{\phi^{\otimes k}}$, as $N\rightarrow\infty$, in trace norm:
\begin{equation}
\label{eq:chaos}
\forall k\in\N, \qquad \lim_{N\rightarrow\infty} \Tr_{1,\ldots,k}\left|\gamma_N^{(k)}-\ket*{\phi^{\otimes k}}\bra*{\phi^{\otimes k}}\right|=0.
\end{equation}

\subsection{Prior results}
The subject of approximating the dynamics of Bose gases in the sense of \eqref{eq:chaos} dates to the 1970s and 1980s through work of Hepp \cite{Hepp1974}, Ginibre and Velo \cite{GV1979_I, GV1979_II}, and Spohn \cite{Spohn80}. After a number of years of inactivity, this subject exerienced a revival in the early 2000s with work of Bardos, Golse, and Mauser \cite{BGM2000}; Fr\"ohlich, Tsai, and Yau \cite{FTY2000}; and Erd\"os and Yau \cite{EY2001}. After a landmark series of works by Erd\"os, Schlein, and Yau \cite{ESY2006, ESY2007, ESY2009, ESY2010}, an explosion of research occurred for the subject of \emph{effective equations} (e.g. mean-field) for quantum many-body systems, with contributions by many authors. As it is not our intention to review this body of literature and we are exclusively interested here in results pertaining to the LL model, we refer the reader to the surveys of Schlein \cite{schlein_clay} and Rougerie \cite{Rougerie2020}, and references therein, for more discussion on this general subject. We intend no offense by any omissions.

The first result on the mean-field approximation for the LL model is due to Adami, Bardos, Golse, and Teta \cite{ABGT2004}. Proceeding by the so-called BBGKY method, which was pioneered by Spohn \cite{Spohn80}, Adami et al. show that for each $k\in\N$ fixed, the sequence $\{\gamma_N^{(k)}\}_{N=1}^\infty$ has a limit point $\gamma^{(k)}$ with respect to a topology weaker than trace norm. They then show that the sequence $\{\gamma^{(k)}\}_{k=1}^\infty$ is a solution to the GP hierarchy \eqref{eq:GP} with initial datum $(\ket*{\phi_0^{\otimes k}}\bra*{\phi_0^{\otimes k}})_{k=1}^\infty$ in a certain class akin to the Sobolev space $H^1$. In order to conclude their proof, they need to show that there can only be one such solution (i.e. prove uniqueness for the GP hierarchy in the class under consideration), from which the convergence \eqref{eq:chaos} follows. However, they could not prove this uniqueness, and to our knowledge, their argument has yet to be completed. We remark that the BBGKY approach does not yield a rate of convergence in \eqref{eq:chaos} as $N\rightarrow\infty$ and $|t| \rightarrow\infty$.

Several years later, Ammari and Breteaux \cite{AmBre2012} revisited the mean-field approximation to the LL model from the perspective of quantum field theory. Inspired by the approach of Rodnianski and Schlein \cite{RS2009}, which in turn builds on earlier ideas of Hepp \cite{Hepp1974} and Ginibre and Velo \cite{GV1979_I, GV1979_II}, the authors use the framework of second quantization and reformulate the problem of mean-field limit for the Hamiltonian \eqref{eq:LL_ham} in terms of the semiclassical limit for a related Hamiltonian on the Fock space. Through a very technical argument involving abstract non-autonomous Schr\"odinger equations, they construct a time-dependent quadratic Hamiltonian which provides a semiclasical approximation for the evolution of coherent states. Borrowing an argument from \cite{RS2009}, they are able to show the convergence \eqref{eq:chaos} from their approximation result for coherent states. We note that the authors do not provide a rate for the convergence \eqref{eq:chaos} in terms of $N$ and $t$.

Lastly, we mention the works \cite{AGT2007, CH2016_1D, NN2019}, which treat the derivation of the 1D cubic NLS from a many-body problem similar to \eqref{eq:LL_ham}, but with the $\d$ potential replaced by a less singular potential of the form $V_N(x)\coloneqq N^{\beta}V(N^\beta x)$, for varying $0<\beta<\infty$.

\subsection{Overview of main results}
\label{ssec:intro_proof}
Having introduced the LL model in the mean-field regime and reviewed prior work, we now state our main results. To the state the theorems, we introduce the notation $\beta_N(\Phi_{N,0},\phi_0)$ for a functional defined in \cref{eq:intro_beta_def} below that measures the initial purity of the condensate. We also introduce the \emph{microscopic energy per particle} and the \emph{NLS energy} respectively given by
\begin{align}
E_N^\Phi &\coloneqq \frac{1}{N}\ip{\Phi_N}{H_N\Phi_N}_{L^2(\R^N)}, \label{eq:intro_epp} \\
E^{\phi} &\coloneqq \|\nabla\phi\|_{L^2(\R)}^2 + \frac{\kappa}{2}\|\phi\|_{L^4(\R)}^4. \label{eq:nls_e}
\end{align}
Our first theorem provides a quantitative rate of convergence to mean-field dynamics in both the repulsive and attractive settings, assuming the limiting state $\phi_0$ is in $H^2(\R)$.

\begin{restatable}[Main result $H^2$]{thm}{mainH}
\label{thm:mainH}
Fix $\kappa\in\{\pm 1\}$, and let $\phi$ be the solution to \eqref{eq:NLS} with initial datum $\phi_0\in H^2(\R)$, such that $\|\phi_0\|_{L^2(\R)}=1$. Let $\Phi_N$ be the solution to \eqref{eq:N_Schr} with initial datum $\Phi_{N,0}\in H^1(\R^N)$ and $\|\Phi_{N,0}\|_{L^2(\R^N)}=1$. There exists an absolute constant $C>0$ such that for every $N\in\N$ and $k\in\{1,\ldots,N\}$,
\begin{equation}
\label{eq:rate}
\begin{split}
&\Tr_{1,\ldots,k}\left|\gamma_N^{(k)}(t) - \ket*{\phi^{\otimes k}}\bra*{\phi^{\otimes k}}(t)\right|\\
&\leq C\sqrt{k}\paren*{\beta_N(\Phi_{N,0},\phi_0)^{1/2}+ |t|^{1/2}\paren*{\frac{\|\phi_0\|_{H^1(\R)}}{N^{1/6}} + \frac{\|\phi_0\|_{H^2(\R)}}{N^{1/4}} + \|\phi_0\|_{H^1(\R)}|E_N^{\Phi}-E^{\phi}|^{1/2}}} e^{C\|\phi_0\|_{H^2(\R)}^2 |t|}
\end{split}
\end{equation}
for every $t\in\R$. In particular, if $\beta_N(\Phi_{N,0},\phi_0)\rightarrow 0$ and $E_N^{\Phi}\rightarrow E^\phi$ as $N\rightarrow\infty$, then mean-field convergence holds.
\end{restatable}

\begin{remark}
\label{rem:H2_ID}
The obvious example of $N$-body initial data to consider in \cref{thm:mainH} is $\Phi_{N,0}=\phi_0^{\otimes N}$, for which it is an easy calculation to show that $\beta_N(\Phi_{N,0},\phi_0) = 0$ and $E_N^{\Phi}-E^{\phi}= O(1/N)$. With a bit more work, one can allow for initial data $\Phi_{N,0}$ with two-body correlations on a length scale vanishing as $N\rightarrow\infty$.
\end{remark}

By further exploiting the $L^2$-subcritical well-posedness of the 1D cubic NLS (see \cref{prop:nls_gwp}), we can show that mean-field convergence still holds assuming only that the limiting one-particle initial state $\phi_0\in L^2(\R)$, provided that we restrict to a smaller class of $N$-body initial data. To our knowledge, this is the first time such a result has been shown. We have not stated the following theorem under the optimal class of $N$-body initial data that can be considered, so as not to overly complicate the statement. But an examination of the argument in \cref{sec:proof_main} will show the interested reader how to allow for more general $N$-body initial data. 

\begin{restatable}[Main result $L^2$]{thm}{mainL}
\label{thm:mainL}
Fix $\kappa\in\{\pm 1\}$, and let $\phi$ be the solution to the NLS \eqref{eq:NLS} with initial datum $\phi_0$, such that $\|\phi_0\|_{L^2(\R)}=1$. Let $\Phi_N$ be the solution to \eqref{eq:N_Schr} with initial datum
\begin{equation}
\Phi_{N,0}\coloneqq  \frac{(P_{\leq (\log N)^{\eta}}\phi_0)^{\otimes N}}{\|P_{\leq (\log N)^{\eta}}\phi_0\|_{L^2(\R)}^N},\footnotemark
\end{equation}
\footnotetext{Here, $P_{\leq M}$ denotes the Littlewood-Paley projector onto frequencies $\lesssim M$.}for fixed $0<\eta<1/4$. There exists an absolute constant $C>0$ such that for every $N\in\N$ and $k\in\{1,\ldots,N\}$,
\begin{equation}
\label{eq:rateL2}
\Tr_{1,\ldots,k}\left|\gamma_N^{(k)}(t) - \ket*{\phi^{\otimes k}}\bra*{\phi^{\otimes k}}(t)\right| \leq C\sqrt{k}\paren*{\frac{t^{1/2}e^{Ct(\log N)^{4\eta}}(\log N)^\eta}{N^{1/6}} + \|P_{>(\log N)^\eta}\phi_0\|_{L^2(\R)}e^{Ct^{5/2}}}
\end{equation}
for every $t\in\R$. In particular, the right-hand side tends to zero as $N\rightarrow\infty$, locally uniformly in time.
\end{restatable}

\begin{remark}
\label{rem:ID_reg}
The $H^2(\R)$ assumption on the initial datum $\phi_0$ in \cref{thm:mainH} is consistent with the regularity assumption of Ammari and Breteaux \cite{AmBre2012}. However, the $L^2(\R)$ assumption in \cref{thm:mainL} substantially improves upon their result by only requiring $\phi_0$ to have finite mass. In particular, \cref{thm:mainL} affirmatively answers the question first investigated by Adami et al. \cite{ABGT2004} of whether one can derive the 1D cubic NLS assuming the limiting state only has finite mass and energy.
\end{remark}

\begin{remark}
\label{rem:reg_Ham}
An examination of the arguments in \cref{sec:MR} and \cref{sec:proof_main} shows that if we replace the Hamiltonian $H_{N}$ in \eqref{eq:LL_ham} with the \emph{regularized Hamiltonian}
\begin{equation}
\label{eq:reg_Ham}
H_{N,\vep} \coloneqq \sum_{i=1}^N -\Delta_i + \frac{\kappa}{N}\sum_{1\leq i<j\leq N}V_\vep(X_i-X_j), \qquad \kappa\in\{\pm 1\},
\end{equation}
where $V$ is a short-range potential satisfying certain regularity conditions and $V_\vep\coloneqq \vep^{-1} V(\vep^{-1}\cdot)$, then choosing $\ep=N^{-\beta}$ for some $\beta\in (0,\infty)$, mean-field convergence to the 1D cubic NLS also holds with an explicit rate of convergence.
\end{remark}

We believe that the ideas behind the proofs of \Cref{{thm:mainH}} and \Cref{{thm:mainL}} will have applications to further studying norm approximations (cf. \cite{NN2019}) and higher-order corrections to mean-field theory (cf. \cite{BPPS2020}) for the Lieb-Liniger and other models involving singular potentials. Furthermore, our results suggest that the strong regularity assumptions imposed on the initial data $\phi_0$ in the derivation of the quantum effective equations, in particular in the use of Pickl's counting method \cite{KP2010, Pickl2010, Pickl2011, Pickl2015, Mitrouskas2017, JP2018, JLP2019, Bossmann2019, BT2019, BPPS2020}, may be lowered substantially by further exploiting the dispersive properties of the mean-field PDE. We intend to investigate the validity of this hypothesis in future work, as well as instances outside of quantum many-body theory where similar analysis may be employed.

\subsection{Road map of the proof}
\label{ssec:intro_RM}
We now comment on the proofs of \cref{thm:mainH} and \cref{thm:mainL} and highlight the major difficulties and differences from existing work. Inspired by the method of Pickl \cite{Pickl2010, Pickl2011, Pickl2015} and the refinement of this method developed by Knowles and Pickl \cite{KP2010} for derivation of the Hartree equation in the mean-field limit, our argument is based on an energy-type estimate for a functional $\beta_N$, which gives a weighted count of the number of ``bad particles'' in the system at time $t$ which are not in the state $\phi(t)$, where $\phi$ solves the cubic NLS \eqref{eq:NLS}. $\beta_N$ takes the form
\begin{equation}
\label{eq:intro_beta_def}
\beta_N(\Phi_N(t),\phi(t)) \coloneqq \ip{\Phi_N(t)}{\wh{n_N(t)}\Phi_N(t)}_{L^2(\R^N)} = \sum_{k=0}^N \sqrt{\frac{k}{N}}\ip{\Phi_N(t)}{P_k(t)\Phi_N(t)}_{L^2(\R^N)},
\end{equation}
where $\Phi_N$ is the solution to \eqref{eq:N_Schr} and $P_k(t)$ is the projector mapping a wave function onto the subspace of $L_{sym}^2(\R^N)$ of functions corresponding to $k$ of the particles not being in the state $\phi(t)$. See \cref{eq:Pk_def} and more generally \cref{ssec:MR_proj} for the precise definition and properties of these projectors. The main estimate for $\beta_N$ is given by \cref{prop:beta_evol}. We defer the precise statement of the proposition to \cref{ssec:MR_beta}, but the estimate controls the evolution of $\beta_N$ in terms of its initial value up to an error vanishing as $N\rightarrow\infty$.

To prove \cref{prop:beta_evol}, we proceed by a Gronwall-type argument. Differentiating $\beta_N$ with respect to time and performing some simplifications, we find that there are three terms we need to estimate, the most difficult of which is
\begin{align}
\mathrm{Term}_3 &\coloneqq \ip{\Phi_{N}}{p_1p_2\comm{(N-1)V_{12}}{\wh{n_N}}q_1q_2\Phi_{N}}_{L_{\ux_N}^2(\R^N)},
\end{align}
where we have used the notation $V_{12}\coloneqq \delta(X_1-X_2)$ and $V_j^\phi \coloneqq |\phi(X_j)|^2$. Here, $p_j$ is the rank-one projector $\ket*{\phi}\bra*{\phi}$ acting in the $x_j$-variable, and $q_j=\id_N-p_j$, where $\id_N$ is the identity operator on $L^2(\R^N)$ (see \cref{ssec:MR_proj} for more details). $V_{12}(q_1q_2\Phi_N)$ and $V_{12}(\wh{n_N}q_1q_2\Phi_N)$, similarly for the other terms, should be interpreted as elements of $H^{-1}(\R^N)$ and the inner product as a duality pairing. By expanding the commutator in the definition of $\mathrm{Term}_3$ and using \cref{lem:shift} to shift the projectors $P_k$ in the definition of $\wh{n_N}$ (see \cref{def:mn_N}), we reduce to bounding the expression
\begin{equation}
\label{eq:intro_T3_red}
\left|\ip{\Phi_N}{p_1p_2V_{12}q_1q_2\wh{\nu_N}\Phi_N}_{L_{\ux_N}^2(\R^N)}\right|,
\end{equation}
where $\wh{\nu_N}=\sum_{k=0}^N \nu_N(k)P_k$ is a time-dependent operator on $L_{sym}^2(\R^N)$ such that the coefficients satisfy $\nu_N(k)\lesssim n_N^{-1}(k)$. See \eqref{eq:nu_def} for the precise definition of $\nu_N$ and $\wh{\nu_N}$.

In \cite{KP2010}, Knowles and Pickl had to contend with an expression similar to $\mathrm{Term}_3$ but with a much more regular potential $V$, which satisfies certain integrability assumptions of the form $V\in L^{p_0}+L^{\infty}$. To deal with their analogue of \eqref{eq:intro_T3_red}, they split the potential into its ``regular'' and ``singular'' parts by making an $N$-dependent decomposition of the form
\begin{equation}
V_{reg} \coloneqq V1_{\{|V|\leq N^{\sigma}\}}, \qquad V_{sing} \coloneqq V1_{\{V|>N^{\sigma}\}},
\end{equation}
where $1_{\{\cdot\}}$ denotes the indicator function for the set $\{\cdot\}$ and $\sigma\in (0,1)$ is a parameter to be optimized at the end. For the singular part, they express the potential as the divergence of a vector field (i.e. $V=\nabla\cdot\xi$) and integrate by parts. Crucially, their integrability assumption implies that $\|\xi\|_{L^2(\R^{3N})}=O(N^{-\delta})$, for some $\delta>0$, which is necessary to close their estimate. For the regular part, the important idea is to exploit the permutation symmetry of the wave function, since the operator norm of $p_1p_2V_{12}q_1q_2$ is much smaller on the bosonic subspace $L_{sym}^2(\R^{3N})$ than on the full space $L^2(\R^{3N})$. As the argument is a bit involved, we only comment that it requires $V_{reg}^2$ to be integrable.

For $V=\delta(x)$, the Knowles-Pickl argument described above breaks down. While we have the identity
\begin{equation}
\label{eq:intro_sgn}
\delta(x) = \frac{1}{2}\nabla\sgn(x),
\end{equation}
the signum function is only in $L^\infty$, not in $L^2$ as their singular-part argument requires. Additionally, since $\delta$ is only a distribution, we cannot assign meaning to $\delta^2$ in the regular part of their argument. In fact, the regular part of their argument is formally vacuous for the $\delta$ potential.

To overcome the difficulties stemming from the lack of integrability of the $\delta$ potential, we introduce a new short-range approximation argument as follows. We make an $N$-dependent mollification of the potential by setting
\begin{equation}
V_\sigma(x)\coloneqq N^{\sigma}\tl{V}(N^{\sigma}x), \qquad \forall x\in\R,
\end{equation}
where $\sigma\in (0,1)$, $0\leq \tl{V}\leq 1$, $\tl{V}\in C_c^\infty(\R)$ is even, and $\int_{\R}dx \tl{V}(x)=1$. By the triangle inequality, we have
\begin{equation}
\label{eq:intro_T3_rem}
\begin{split}
\left|\ip{\Phi_N}{p_1p_2V_{12}q_1q_2\wh{\nu_N}\Phi_N}_{L_{\ux_N}^2(\R^N)}\right| &\leq \left|\ip{\Phi_N}{p_1p_2(V_{12}-V_{\sigma,12})q_1q_2\wh{\nu_N}\Phi_N}_{L_{\ux_N}^2(\R^N)}\right| \\
&\ph+ \left|\ip{\Phi_N}{p_1p_2V_{\sigma,12}q_1q_2\wh{\nu_N}\Phi_N}_{L_{\ux_N}^2(\R^N)}\right|.
\end{split}
\end{equation}
Combining the scaling relation
\begin{equation}
\int_{\R}dx |x|^{1/2}V_\sigma(x) \sim N^{-\sigma/2}
\end{equation}
with fact that the wave function $\Phi_N$ is $\frac{1}{2}$-H\"older-continuous in a single particle variable by conservation of mass and energy together with Sobolev embedding (see \cref{lem:H1_hold}), we can  show that the first term in the right-hand side of \eqref{eq:intro_T3_rem} is controlled by $\beta_N$ up to a small error vanishing as $N\rightarrow\infty$.  We can now estimate the second term in the right-hand side of \eqref{eq:intro_T3_rem} by proceeding similarly as to the aforementioned Knowles-Pickl argument for the regular part $V_{reg}$ of the potential. While $\|V_\sigma\|_{L^2(\R)}\sim N^{\sigma/2}$, we are able to extract sufficient decay in $N$ from other factors to absorb this growth, provided we appropriately choose $\sigma$.

This mollification idea seems quite powerful, and we expect it to have further application to problems of mean-field convergence for both quantum and classical interacting particle systems. Indeed, in the recent work \cite{Rosenzweig2020_MFPV, Rosenzweig2020_MFCou} by the author, a similar, but more complicated time-dependent version of our mollification argument is used to prove mean-field convergence of the Helmholtz-Kirchoff point vortex system to the 2D incompressible Euler equation with vorticity in the scaling-critical $L^\infty$ space. 

To extend the above analysis to $\phi_0\in L^2(\R)$, as in \cref{thm:mainL}, we introduce another new idea, which is to exploit the quantitative dependence on the initial data for the mean-field equation itself (i.e. the cubic NLS).\footnote{The present idea is to our knowledge novel, but it is worth mentioning that the use of methods inspired by the area of dispersive nonlinear PDE dates to the important work of Klainerman and Machedon \cite{KM2008}. A non-exhaustive sample of applications of such methods may be found in the subsequent works \cite{GMM2010, KSS2011, GMM2011, GM2013, CP2014, CHPS2015, Sohinger2015, Chong2016, GM2017, CH2019} and references therein.} The crucial observation for this step is that in the statement of \cref{prop:beta_evol}, $\phi_0\in H^2(\R)$ can be arbitrary and similarly for $\Phi_{N,0}\in H^1(\R^N)$. Therefore, we have the freedom to \emph{mollify} the solution $\phi$ to \eqref{eq:NLS} by mollifying the initial datum $\phi_0$, so that it is now in $H^2(\R)$. This mollification is most easily accomplished by a high-frequency cut-off $\phi_{N,0}$ of the initial datum, leading to
\begin{equation}
\label{eq:nls_moll}
\begin{cases}
(i\p_t+\D)\phi_N = \kappa|\phi_N|^2\phi_N \\
\phi_{N}|_{t=0} = \phi_{N,0} \coloneqq \frac{P_{\leq \rho(N)}\phi_0}{\|P_{\leq\rho(N)}\phi_0\|_{L^2(\R)}}
\end{cases}, \qquad \kappa\in\{\pm1\}.
\end{equation}
Here, $P_{\leq \rho(N)}$ is the Littlewood-Paley projector to frequencies $\lesssim \rho(N)$ and $\rho$ is a suitable rate function tending to $\infty$ as $N\rightarrow\infty$. Restricting to the $1$-particle marginal, we have by the triangle inequality that
\begin{equation}
\label{eq:TI_intro}
\begin{split}
\Tr\left|\gamma_N^{(1)}-\ket*{\phi}\bra*{\phi}\right| \leq \Tr\left|\gamma_N^{(1)}-\ket*{\phi_N}\bra*{\phi_N}\right| + \Tr\left|\ket*{\phi_N}\bra*{\phi_N}-\ket*{\phi}\bra*{\phi}\right|.
\end{split}
\end{equation}
Since $\phi_{N,0}\in H^2(\R)$ with unit $L^2$ norm, we can apply \cref{prop:beta_evol} to the first term. While
\begin{equation}
\|\phi_N\|_{L_t^\infty H_x^2(\R\times\R)} \lesssim\|\phi_{N,0}\|_{H^2(\R)} \lesssim \rho(N)^2,
\end{equation}
we can absorb this growth in $N$ by appropriately choosing $\rho(N)$ so that it does not approach infinity too quickly. We can estimate the second term in \eqref{eq:TI_intro} by using the dependence on initial data estimate of \cref{prop:nls_gwp}, thereby completing the argument.

\subsection{Organization of the paper}
We now comment on the organization of the paper. \cref{sec:pre} is devoted to basic notation, preliminary facts from functional analysis, and the well-posedness theory of the 1D cubic NLS. We begin the section with an index of the frequently used notation in the article (see \cref{tab:notation}). In \cref{sec:MR}, we prove \cref{prop:beta_evol}, which is the main estimate for the functional $\beta_N$ and the main ingredient for the proofs of \Cref{thm:mainH} and \Cref{thm:mainL}. As this section constitutes the bulk of the paper, we have divided it into several subsections beginning with background material for projection operators in \cref{ssec:MR_proj} and continuing with the steps in the proof of \cref{prop:beta_evol}, the statement of which is given in \cref{ssec:MR_beta}. Lastly, in \cref{sec:proof_main}, we show how to obtain \Cref{thm:mainH} and \Cref{thm:mainL} from \cref{prop:beta_evol}. 

\subsection{Acknowledgments}
The author thanks Lea Bo{\ss}mann and Nata\v{s}a Pavlovi\'c for discussion which inspired him to revisit the subject of this article and Peter Pickl for helpful correspondence regarding his method. The author also thanks Dana Mendelson, Andrea R. Nahmod, and Gigliola Staffilani for numerous discussions on the exact solvability of the LL model and its connection to the integrability of the cubic NLS, which have informed the presentation of this article. Lastly, the author thanks Avy Soffer for his encouraging and engaging conversation related to this project. The author gratefully acknowledges financial support from the University of Texas at Austin and the Simons Collaboration on Wave Turbulence.

\section{Preliminaries}
\label{sec:pre}
We include here a table of the notation frequently used in the article with an explanation for the notation and/or a reference to where the definition is given.

\begin{table}[h]
\label{tab:not}
   \caption{Notation} 
   \label{tab:notation}
   \small 
   \centering 
   \begin{tabular}{|p{3cm}|p{14cm}|} 
      \hline
   \textbf{Symbol} & \textbf{Definition} \\ 
      \hline
	$A\lesssim B,\ A\sim B$ & There are absolute constants $C_1,C_2>0$ such that $A\leq C_1 B$ or $C_2B\leq A\leq C_1 B$ \\
   $\ul{x}_{k}, \ \ul{x}_{i;i+k}$ & $(x_1, \ldots, x_k), \ (x_i, \ldots, x_{i+k})$, where $x_j\in\R$ for $j\in\{1,\ldots,k\}$ or $j\in\{i,\ldots,i+k\}$\\
      $d\ul{x}_{k}, \ d\ul{x}_{i;i+k}$ & $ dx_1 \cdots dx_k, \ dx_i \cdots dx_{i+k}$ \\
         $\N, \ \N_0$ & natural numbers, \ natural numbers inclusive of zero\\
      $L^p(\R^N), \ \|\cdot\|_{L^p}$ & standard $p$-integrable function space: see \eqref{eq:L^p_def}\\
      $H^s(\R^N), \ \|\cdot\|_{H^s}$ & standard $L^2$-based Sobolev function space: see \eqref{eq:H^s_def}\\
      $C^\gamma(\R^N), \ \|\cdot\|_{C^\gamma}$ & standard H\"older-continuous function space: see \eqref{eq:C^gam_def}\\
      $\ip{\cdot}{\cdot}$ & $L^2(\R^N)$ inner product with physicist's convention: $\ip{f}{g} \coloneqq \int_{\R^N}d\ux_N \ol{f(\ux_N)}g(\ux_N)$\\
      $\ipp{\cdot,\cdot}$ & duality pairing\\
      $\bra*{\cdot} \ \ket*{\cdot}$ & Dirac's bra-ket notation: see footnote \ref{fn:Dirac}\\
       $A_{i_1\cdots i_k}^{(k)}$ & subscript denotes that the operator on $L^2(\R^N)$ acts in the variables $(x_{i_1},\ldots,x_{i_k})$ \\
       $\phi^{\otimes k}$ & $k$-fold tensor product of $\phi$ with itself realized as $\phi^{\otimes k}(\ux_k) = \prod_{i=1}^k\phi(x_i), \ \ux_k\in\R^k$\\
        $\Tr_{1,\ldots,N}$ & trace on $L^2(\R^N)$ \\
        $\Tr_{k+1,\ldots,N}$ & partial trace on $L^2(\R^N)$ over $x_{k+1},\ldots,x_N$ coordinates\\
        $\id, \ \id_N$ & identity operator on $L^2(\R)$ and on $L^2(\R^N)$\\
        $H_{N}, \ H_{N,\vep}$ & LL Hamiltonian and regularized LL Hamiltonian: see \eqref{eq:LL_ham} and \eqref{eq:reg_Ham}\\
        $p_j, \ q_j$ & projectors $\id^{\otimes j-1}\otimes p \otimes \id^{N-j}, \ \id^{\otimes j-1} \otimes q \otimes \id^{N-j}$: see \eqref{eq:pj_qj_def}\\
        $P_k$ & projector onto subspace of $k$ particles not in the state $\phi(t)$: see \eqref{eq:Pk_def} \\
        $\wh{f}, \ \wh{f}^{-1}$ & operator $L^2(\R^N)\rightarrow L^2(\R^N)$ defined by $\wh{f}\coloneqq \sum_{k=0}^N f(k)P_k$, for $f:\Z\rightarrow \C$: see \eqref{eq:hat_def}\\
        $n_N, m_N$ \ $\wh{n_N}, \wh{m_N}$ & functions $\Z\rightarrow\C$ and operators $L^2(\R^N)\rightarrow L^2(\R^N)$: see \cref{def:mn_N} \\
        $\mu, \nu$ \ $\wh{\mu}, \wh{\nu}$ & functions $\Z\rightarrow\C$ and operators $L^2(\R^N)\rightarrow L^2(\R^N)$: see \eqref{eq:mu_def} and \eqref{eq:nu_def}\\
        $\alpha_N, \ \beta_N$ & time-dependent functionals of solution $\phi$ to \eqref{eq:NLS} and $\Phi_N$ to \eqref{eq:N_Schr}: see \cref{def:mn_N}\\
        $\tau_n$ & shift operator on $\C^{\Z}$: see \eqref{eq:shift_def}\\
        $\uD_k$ & Laplacian on $\R^k$: $\uD_k \coloneqq \sum_{i=1}^k \Delta_i$\\
        $\comm{\cdot}{\cdot}$ & commutator bracket: $\comm{A}{B} \coloneqq AB-BA$\\
        \hline
   \end{tabular}
\end{table}

\subsection{Function spaces}
\label{ssec:pre_fs}
Fix $N\in\N$. We denote the Schwartz space on $\R^N$ by $\Sc(\R^N)$ and the dual space of tempered distributions on $\R^N$ by $\Sc'(\R^N)$. The subspace of $\Sc(\R^N)$ consisting of functions with compact support is denoted by $C_c^\infty(\R^N)$. Given a Schwartz function $\Phi\in\Sc(\R^N)$ and a tempered distribution $\Upsilon\in\Sc'(\R^N)$, we denote their duality pairing by
\begin{equation}
\ipp{\Phi,\Upsilon}_{\Sc(\R^N)-\Sc'(\R^N)} \coloneqq \Upsilon(\Phi).
\end{equation}
For $1\leq p\leq\infty$, we define $L^p(\R^N)$ to be the usual Banach space of equivalence classes of measurable functions $\Phi:\R^N\rightarrow\C$ with respect to the norm
\begin{equation}
\label{eq:L^p_def}
\|\Phi\|_{L^p(\R^N)} \coloneqq \paren*{\int_{\R^N}d\ux_N |\Phi(\ux_N)|^p}^{1/p}
\end{equation}
with obvious modification when $p=\infty$. We denote the inner product on $L^2(\R^N)$ by
\begin{equation}
\ip{\Phi}{\Psi}_{L^2(\R^N)} \coloneqq \int_{\R^N}d\ux_N \ol{\Phi(\ux_N)}\Psi(\ux_N).
\end{equation}
Note that we use the physicist's convention that the inner product is complex linear in the second entry. For $s\in\R$, we define the Sobolev space $H^s(\R^N)$ to be the completion of the space $\Sc(\R^N)$ with respect to the norm
\begin{equation}
\label{eq:H^s_def}
\|\Phi\|_{H^s(\R^N)} \coloneqq \paren*{\int_{\R^N}d\uxi_N \jp{\uxi_N}^{2s}|\F(\Phi)(\uxi_N)|^2}^{1/2},
\end{equation}
where $\mathcal{F}$ denotes the Fourier transform and $\jp{x}\coloneqq (1+|x|^2)^{1/2}$ is the Japanese bracket. Evidently, we can anti-isomorphically identify $H^{-s}(\R^N)$ with the dual space $(H^s(\R^N))^*$. For $\gamma\in (0,1)$, we denote the H\"older norm on $\R^N$ of exponent $\gamma$ by
\begin{equation}
\label{eq:C^gam_def}
\|\Phi\|_{\dot{C}^\gamma(\R^N)} \coloneqq  \sup_{{x,y\in\R^N} \atop {x\neq y}} \frac{|\Phi(x)-\Phi(y)|}{|x-y|^{\gamma}}, \qquad \|\Phi\|_{C^\gamma(\R^N)} \coloneqq \|\Phi\|_{L^\infty(\R^N)} + \|\Phi\|_{\dot{C}^\gamma(\R^N)}.
\end{equation}

\begin{remark}
In the sequel, we generally omit the underlying domain for norms (e.g. we write $\|\cdot\|_{L^p}$ instead of $\|\cdot\|_{L^p(\R^N)}$), as the domain will be clear from context. Similarly, we omit the underlying domain for the inner product $\ip{\cdot}{\cdot}$ and for the duality pairing $\ipp{\cdot,\cdot}$. To avoid any confusion, we generally reserve upper-case Greek letters (e.g. $\Phi,\Psi$) for functions or distributions $\R^N\rightarrow\C$ and lower-case Greek letters (e.g. $\varphi,\psi$) for functions or distributions $\R\rightarrow\C$. To emphasize the variable with respect to which a norm is taken, we use a subscript (e.g. $C_t^0$, $L_x^2$, or $L_{\ux_N}^2$).
\end{remark}

We record here a partial H\"older continuity result for functions in $H^1(\R^N)$ used in \cref{sec:MR} for the $N$-body wave function $\Phi_N$.
\begin{lemma}[Partial H\"older continuity]
\label{lem:H1_hold}
For $N\in\N$ and any $i\in\{1,\ldots,N\}$, we have the estimate
\begin{equation}
\|\Phi\|_{L_{(\ux_{1;i-1},\ux_{i+1;N})}^2(\R^{N-1};\dot{C}_{x_i}^{1/2}(\R))} \leq \|\nabla_i\Phi\|_{L^2(\R^N)}, \qquad \Phi\in\Sc(\R^N).
\end{equation}
Consequently, every element of $H^1(\R^N)$ has a modification belonging to $L_{(\ux_{1;i-1},\ux_{i+1;N})}^2(\R^{N-1};C_{x_i}^{1/2}(\R))$.
\end{lemma}
\begin{proof}
By considerations of symmetry, it suffices to consider $i=1$. Let $\Phi\in\Sc(\R^N)$, and fix $\ux_{2;N}\in\R^{N-1}$. Define the function
\begin{equation}
\phi_{\ux_{2;N}}: \R\rightarrow\C, \qquad \phi_{\ux_{2;N}}(x) \coloneqq \Phi(x,\ux_{2;N}), \enspace \forall x\in\R.
\end{equation}
Applying the fundamental theorem of calculus to $\phi_{\ux_{2;N}}$ followed by Cauchy-Schwarz, we obtain that
\begin{equation}
|\phi_{\ux_{2;N}}(x)-\phi_{\ux_{2;N}}(y)| \leq |x-y|^{1/2}\|\nabla\phi_{\ux_{2;N}}\|_{L^2(\R)}, \qquad \forall x,y\in\R,
\end{equation}
which implies that $\|\phi_{\ux_{2;N}}\|_{\dot{C}^{1/2}(\R)} \leq \|\nabla\phi_{\ux_{2;N}}\|_{L^2(\R)}$. 
Therefore, we see from the Fubini-Tonelli theorem that
\begin{align}
\int_{\R^{N-1}}d\ux_{2;N} \|\phi_{\ux_{2;N}}\|_{\dot{C}^{1/2}(\R)}^2 &\leq \int_{\R^{N-1}}d\ux_{2;N} \|\nabla\phi_{\ux_{2;N}}\|_{L^2(\R)}^2  =\|\nabla_1\Phi\|_{L^2(\R^N)}^2.
\end{align}
The conclusion of the proof then follows from the density of $\Sc(\R^N)\subset H^1(\R^N)$.
\end{proof}

\subsection{The 1D cubic NLS}
\label{ssec:pre_nls}
We recall some basic facts from the well-posedness theory--in particular the $L^2$-subcritical nature of it--for the 1D cubic NLS \eqref{eq:NLS} that we shall use to prove \cref{thm:mainH} and \cref{thm:mainL} in \cref{sec:proof_main}. The material presented here may be found with more details in \cite[Chapters 2 and 3]{Tao2006}.

\begin{mydef}[Strichartz norm]
For $2\leq p,q\leq \infty$, we say that the pair $(p,q)$ is \emph{(Schr\"odinger) Strichartz admissible} if
\begin{equation}
\frac{2}{p} = \frac{1}{2}-\frac{1}{q}.
\end{equation}
For a interval $I\subset\R$, we define the \emph{Strichartz space} $S^0(I\times\R)$ to be the closure $\Sc(\R\times\R)$ under the norm
\begin{equation}
\label{eq:Strich_norm}
\|\phi\|_{S^0(I\times\R)} \coloneqq \sup_{\text{$(p,q)$ admissible}} \|\phi\|_{L_t^p L_x^q(I\times\R)}.
\end{equation}
We define $N^0(I\times\R)$ to be the dual norm.
\end{mydef}

\begin{prop}
\label{prop:nls_gwp}
For any $\phi_0\in L^2(\R)$, there exists a unique, global solution $\phi \in C(\R; L^2(\R))$
in the sense that for any finite $T>0$, $\|\phi\|_{S^0([-T,T]\times\R)}<\infty$ and $\phi$ satisfies the Duhamel formula
\begin{equation}
\phi(t) = e^{it\D}\phi_0 - i\kappa\int_0^t e^{i(t-\tau)\D}(|\phi(\tau)|^2\phi(\tau))d\tau, \qquad  t\in [-T,T].
\end{equation}
and the Strichartz norm growth bound
\begin{equation}
\label{eq:Strich_bbd}
\|\phi\|_{S^0([-T,T]\times\R)} \lesssim T\|\phi_0\|_{L^2(\R)}^5.
\end{equation}
Moreover, the solution depends Lipschitz continuously on the initial data: if $\phi$ and $\psi$ are two solutions, then
\begin{equation}
\label{eq:dep_bnd}
\|\phi-\psi\|_{L_t^\infty L_x^2([0,T]\times\R)} \lesssim \|\phi(0)-\psi(0)\|_{L^2(\R)}e^{CT^{1/2}(\|\phi\|_{L_t^4 L_x^\infty([0,T]\times\R)}^2 + \|\psi\|_{L_t^4 L_x^\infty([0,T]\times\R)}^2)},
\end{equation}
where $C>0$ is an absolute constant.
\end{prop}
\begin{proof}
We only sketch the proofs of the estimates \eqref{eq:Strich_bbd}, \eqref{eq:dep_bnd}. By rescaling the solution through
\begin{equation}
\phi_\lambda(t,x) \coloneqq \lambda\phi(\lambda^2 t,\lambda x) \qquad \lambda = \|\phi_0\|_{L^2(\R)}^{-2},
\end{equation}
we may assume without loss of generality that $\phi$ has unit mass. It follows from Duhamel's formula, H\"older's inequality, and Strichartz estimates that
\begin{equation}
\label{eq:Duh_bnd}
\|\phi\|_{S^0([-T,T]\times\R)} \leq C + T^{1/2} \|\phi\|_{S^0([-T,T]\times\R)}^3,
\end{equation}
where $C>0$ is some absolute constant. So the minimal time $T_0>0$ such that $\|\phi\|_{S^0([-T_0,T_0]\times\R)} = 2C$ must satisfy the lower bound $T_0\gtrsim 1$. Now given an interval $[0,T]$, for $T>0$, (by time reversal symmetry the case $[-T,0]$ will follow from our argument), we chop it into $\sim T/T_0$ subintervals of length $\sim T_0$. By using conservation of mass and iterating the argument using \eqref{eq:Duh_bnd} on each subinterval, we find that
\begin{equation}
\|\phi\|_{S^0([-T,T]\times\R)} \lesssim T.
\end{equation}

For the dependence estimate, it follows from subtracting the Duhamel formulae for $\phi$ and $\psi$, then applying Strichartz estimates, followed by H\"older's inequality that
\begin{align}
\|\phi(t)-\psi(t)\|_{L^2(\R)} &\leq \|\phi(0)-\phi(0)\|_{L^2(\R)} + C\int_0^t \|\phi(s)-\psi(s)\|_{L^2(\R)}\paren*{\|\phi(s)\|_{L^\infty(\R)}^2 + \|\psi(s)\|_{L^\infty(\R)}^2}ds .
\end{align}
By the Gronwall-Bellman inequality followed by Cauchy-Schwarz,
\begin{align}
\|\phi(t)-\psi(t)\|_{L^2(\R)} &\leq \|\phi(0)-\psi(0)\|_{L^2(\R)}e^{C\int_0^t (\|\phi(s)\|_{L^\infty(\R)}^2 + \|\psi(s)\|_{L^\infty(\R)}^2)ds} \nn\\
&\leq \|\phi(0)-\psi(0)\|_{L^2(\R)}e^{Ct^{1/2}(\|\phi\|_{L_s^4L_x^\infty([0,t]\times\R)}^2 + \|\psi\|_{L_s^4L_x^\infty([0,t]\times\R)}^2)},
\end{align}
which yields \eqref{eq:dep_bnd}.
\end{proof}

The global existence in \cref{prop:nls_gwp} is a consequence of the $L^2$-subcritical nature of the local theory for the equation (i.e. the time of existence depends on $\|\phi_0\|_{L^2(\R)}$) and conservation of mass. In addition to conservation of energy and momentum, the 1D cubic NLS has infinitely many conserved quantities, a consequence of its integrability by the inverse scattering transform. In fact, Koch and Tataru \cite{KT2018} and Killip, Visan, and Zhang \cite{KVZ2018} have shown that for each $s>-1/2$, there is a conserved quantity controlling the $H^s(\R)$ norm of the solution.

\begin{remark}
\label{eq:KV_gwp}
By heavily exploiting the integrability of the equation, Harrop-Griffiths, Killip, and Visan \cite{HKV2020} have recently shown that the NLS is globally well-posed in $H^{s}(\R)$, for any $s>-1/2$, in the sense that the solution map extends uniquely from the space $\Sc(\R)$. Note that $s=-1/2$ is the scaling-critical regularity for the 1D cubic NLS.  
\end{remark}

\section{The counting functional $\beta_N$}
\label{sec:MR}
\subsection{Projectors}
\label{ssec:MR_proj}
We first define the projectors underlying the definition of the functional $\beta_N$ in the statement of the proposition. For $\phi\in L^2(\R)$, we define the projectors
\begin{equation}
\label{eq:pq_def}
p^\phi \coloneqq \ket*{\phi}\bra*{\phi}, \qquad q^\phi \coloneqq \id - p^\phi,
\end{equation}
where $\id$ denotes the identity operator on $L^2(\R)$. For $N\in\N$ and $j\in\{1,\ldots,N\}$, we define
\begin{equation}
\label{eq:pj_qj_def}
p_j^\phi \coloneqq \id^{\otimes j-1} \otimes p^\phi\otimes \id^{\otimes N-j}, \qquad q_j^\phi \coloneqq \id_N-p_j^\phi = \id^{\otimes j-1}\otimes q^\phi \otimes \id^{\otimes N-j},
\end{equation}
where $\id_N=\id^{\otimes N}$ denotes the identity operator on $L^2(\R^N)$. Since $\id = p^\phi+q^\phi$, it follows that
\begin{equation}
\id_N = (p_1^\phi+q_1^\phi)\cdots (p_N^\phi+q_N^\phi),
\end{equation}
and therefore
\begin{equation}
\label{eq:Pk_def}
\id_N = \sum_{k=0}^N P_k^\phi, \qquad P_k^\phi\coloneqq \sum_{ {\ua_N \in\{0,1\}^N}\atop {|\ua_N|=k}} \prod_{j=1}^N {(p_j^{\phi})}^{1-\alpha_j}{(q_j^{\phi})}^{\alpha_j}.
\end{equation}
We define $P_k^\phi$ to be the zero operator on $L^2(\R^N)$ for $k\in\Z\setminus\{0,\ldots,N\}$. Important properties of the operators $P_k^\phi$ are the following:
\begin{enumerate}[(i)]
\item
$P_k^\phi$ is an orthogonal projector on $L^2(\R^N)$;
\item
$P_k^\phi (L_{sym}^2(\R^N)) \subset L_{sym}^2(\R^N)$;
\item
$P_k^\phi P_l^\phi=\delta_{kl}P_k^\phi$, where $\delta_{kl}$ is the Kronecker delta function;
\item
$p_j^\phi,q_j^\phi$ commute with $P_k^\phi$, for any $j\in\{1,\ldots,N\}$ and $k\in\Z$.
\end{enumerate}
To avoid cumbersome notation, we shall now drop the superscript $\phi$ in the projectors; but the reader should always keep in mind the implicit dependence on $\phi$.

\begin{remark}
In the sequel, we frequently use without comment the elementary fact that $p_j,q_j$ are self-adjoint and that we have the operator norm identities
\begin{equation}
\|p_j\|_{L_{\ux_N}^2(\R^N)\rightarrow L_{\ux_N}^2(\R^N)} = \|q_j\|_{L_{\ux_N}^2(\R^N)\rightarrow L_{\ux_N}^2(\R^N)} =1.
\end{equation}
\end{remark}

Given a function $f:\Z\rightarrow\C$, we define the operator
\begin{equation}
\label{eq:hat_def}
\wh{f} \coloneqq \sum_{k\in\Z} f(k)P_k = \sum_{k=0}^N f(k)P_k.
\end{equation}
The reader may check that for $f,g:\Z\rightarrow\C$, we have that $\wh{fg} = \wh{f}\wh{g}$. Furthermore, since $p_j, q_j, P_k$ commute, it follows that $\wh{f}$ commutes with $p_j, q_j, P_k$. Additionally, if $f,g$ are such that $f\geq g$. Then $\wh{f}\geq \wh{g}$. If $f\geq 0$, then we agree to abuse notation by writing
\begin{equation}
f^{-1}(k)\coloneqq \frac{1}{f(k)}1_{>0}(f(k)) \quad \text{and} \quad \wh{f}^{-1} \coloneqq \sum_{k\in\Z} f^{-1}(k)P_k
\end{equation}
with the convention that $0\cdot\infty=0$.

\begin{mydef}[Counting functions]
\label{def:mn_N}
Given $N\in\N$, we define the functions $m_N, n_N:\Z\rightarrow[0,\infty)$ by
\begin{align}
m_N(k)\coloneqq \frac{k}{N}1_{\geq 0}(k) \quad \text{and} \quad n_N(k)\coloneqq \sqrt{\frac{k}{N}}1_{\geq 0}(k), \qquad \forall k\in\Z.
\end{align}
With the notation introduced in \eqref{eq:hat_def}, we define the quantities
\begin{align}
\alpha_N(\Phi_N,\phi) \coloneqq \ip{\Phi_N}{\wh{m_N}\Phi_N}_{L^2(\R^N)}, \label{eq:alpha_def} \\
\beta_N(\Phi_N,\phi) \coloneqq \ip{\Phi_N}{\wh{n_N}\Phi_N}_{L^2(\R^N)}.\label{eq:beta_def}
\end{align}
If $\Phi_N(t), \phi(t)$ are time-dependent, then we agree to use the compact notation $\alpha_N(t)$ and $\beta_N(t)$.
\end{mydef}

\begin{remark}
\label{rem:alpha_N}
Since $\sum_{k=0}^N P_k = \id_N$, we have that
\begin{equation}
\frac{1}{N}\sum_{j=1}^N q_j = \frac{1}{N}\sum_{k\in\Z}\sum_{j=1}^N q_j P_k.
\end{equation}
By unpacking the definition of $P_k$ in \eqref{eq:Pk_def}, the reader can check that $\sum_{j=1}^N q_j P_k = kP_k$, which implies that
\begin{equation}
\frac{1}{N}\sum_{j=1}^N q_j = \sum_{k\in\Z} \frac{k}{N}P_k = \wh{m_N}.
\end{equation}
It then follows from the symmetry of the wave function $\Phi_N$ under exchange of particle labels that
\begin{equation}
\al_N = \frac{1}{N}\sum_{i=1}^N \ip{\Phi_N}{q_i\Phi_N} = \ip{\Phi_N}{q_1\Phi_N}.
\end{equation}
\end{remark}

We record two technical lemmas from \cite{KP2010} of frequent use in \cref{sec:MR}.
\begin{lemma}[{\cite[Lemma 3.9]{KP2010}}]
\label{lem:mN_est}
For any function $f:\Z\rightarrow [0,\infty)$, the following hold:
\begin{enumerate}[(i)]
\item
\label{item:mN_q1}
\begin{equation}
\|\wh{f}^{1/2}q_1\Phi_N\|_{L_{\ux_N}^2}^2=\ip{\Phi_N}{\wh{f}q_1 \Phi_N}_{L_{\ux_N}^2}= \ip{\Phi_N}{\wh{f}\wh{m_N}\Phi_N}_{L_{\ux_N}^2},
\end{equation}
\item
\label{item:mN_q1q2}
\begin{equation}
\|\wh{f}^{1/2}q_1q_2\Phi_N\|_{L_{\ux_N}^2}^2 = \ip{\Phi_N}{\wh{f}q_1 q_2\Phi_N}_{L_{\ux_N}^2} \leq \frac{N}{N-1} \ip{\Phi_N}{\wh{f}\wh{m_N}^2\Phi_N}_{L_{\ux_N}^2}.
\end{equation}
\end{enumerate}
\end{lemma}
Given $n\in\N$, we define the shift operator
\begin{equation}
\label{eq:shift_def}
\tau_n: \C^{\Z} \rightarrow \C^{\Z}, \qquad (\tau_n f)(k) \coloneqq f(k+n), \quad \forall k\in\Z, \ f\in \C^{\Z}.
\end{equation}
\begin{lemma}[{\cite[Lemma 3.10]{KP2010}}]
\label{lem:shift}
Let $r\in\N$, and let $A^{(r)}$ be a linear operator on $L_{sym}^2(\R^r)$. For $i\in\{1,2\}$, let $Q_i$ be a projector of the form
\begin{equation}
Q_i = \#_1 \cdots \#_r,
\end{equation}
where each $\#$ stands for either $p$ or $q$. Define the linear operator $A_{1\cdots r}^{(r)} \coloneqq A^{(r)}\otimes \id^{N-r}$. Then for any function $f:\Z\rightarrow\C$, we have that 
\begin{equation}
Q_1 A_{1\cdots r}^{(r)} \wh{f} Q_2 = Q_1 \wh{(\tau_n f)} A_{1\cdots r}^{(r)} Q_2,
\end{equation}
where $n\coloneqq n_2-n_1$ and $n_i$ is the number of factors $q$ in $Q_i$, for $i\in \{1,2\}$.
\end{lemma}

\subsection{Estimate for $\beta_N$}
\label{ssec:MR_beta}
The workhorse of this article is the following proposition giving an estimate for the evolution of the functional $\beta_N$. The reader will recall that $E_N^{\Phi_N}$ denotes the microscopic energy per particle \eqref{eq:intro_epp} and $E^\phi$ denotes the NLS energy \eqref{eq:nls_e}.

\begin{restatable}[Evolution of $\beta_N$]{prop}{betaevol}
\label{prop:beta_evol}
Let $\kappa\in\{\pm 1\}$. There exists an absolute constant $C>0$ such that the following holds. Let $\phi$ be a solution to \eqref{eq:NLS} with initial datum $\phi_0$, and let $\Phi_N$ be a solution to \eqref{eq:LL_ham} with initial datum $\Phi_{N,0}$. Then for every $N\in\N$,
\begin{equation}
\beta_N(\Phi_N(t),\phi(t)) \leq \paren*{\beta_N(\Phi_{N,0},\phi_0) + C |t| \paren*{\frac{\|\phi_0\|_{H^1(\R)}^2}{N^{1/3}} + \frac{\|\phi_0\|_{H^2(\R)}^2}{N^{1/2}} +(E_N^{\Phi}-E^{\phi})\|\phi_0\|_{H^1(\R)}^2}}  e^{C\|\phi_0\|_{H^2(\R)}^2 |t|}.
\end{equation}
\end{restatable}
Rather than prove \cref{prop:beta_evol} directly, we prove a similar estimate for the approximation $\beta_{N,\vep}$ defined in \eqref{eq:beta_ep_def} below. The motivation is largely to avoid awkward notation involving distributions and that the validity of \cref{rem:reg_Ham} will become clear from our estimate for $\beta_{N,\vep}$ and the analysis in \cref{sec:proof_main}. Similarly to \eqref{eq:alpha_def} and \eqref{eq:beta_def}, we define 
\begin{align}
\alpha_{N,\vep}(t) &\coloneqq \alpha_N(\Phi_N^{\vep}(t),\phi(t)), \label{eq:alpha_ep_def} \\
\beta_{N,\vep}(t) &\coloneqq \beta_N(\Phi_N^\vep(t),\phi(t)), \label{eq:beta_ep_def}
\end{align}
where $\Phi_N^\vep$ is the solution to the regularized Schr\"odinger equation obtained by replacing $H_N$ in \eqref{eq:N_Schr} with $H_{N,\vep}$ defined in \cref{rem:reg_Ham}. Using the norm-resolvent convergence of $H_{N,\vep}$ to $H_N$ (see \cite[Theorem VII.25]{RSII}) and the following lemma, one can show that $\alpha_{N,\vep}\rightarrow\alpha_N$ and $\beta_{N,\vep}\rightarrow\beta_N$, as $\vep\rightarrow0^+$, uniformly on compact intervals on time. We leave the proof as a simple exercise for the reader.

\begin{lemma}
\label{lem:beta_ep_conv}
Let $T>0$, and let $f:\Z\rightarrow \C$ be bounded. For $N\in\N$ and $\vep>0$, define the functions $\vartheta_N, \vartheta_{N,\vep}:\R\rightarrow \C$ by
\begin{equation}
\vartheta_N(t) \coloneqq \ip{\Phi_N(t)}{\wh{f}(t)\Phi_N(t)}_{L_{\ux_N}^2} \quad \text{and} \quad \vartheta_{N,\vep}(t)\coloneqq \ip{\Phi_N^\vep(t)}{\wh{f}(t)\Phi_N^\vep(t)}_{L_{\ux_N}^2}, \qquad \forall t\in\R.
\end{equation}
Then for $N$ fixed,
\begin{equation}
\lim_{\vep\rightarrow 0^+} \sup_{|t|\leq T} \left|\vartheta_{N,\vep}(t)-\vartheta(t)\right| =0.
\end{equation}
\end{lemma}

\begin{prop}[Evolution of $\beta_{N,\vep}$]
\label{prop:beta_ep_ineq}
For $\kappa\in\{\pm 1\}$, we have the estimate
\begin{equation}
\begin{split}
\dot{\beta}_{N,\vep}(t) &\lesssim \frac{\|\phi(t)\|_{L^\infty(\R)}^2}{N} + \frac{1}{N^{\sigma}} + \frac{\|\phi(t)\|_{L^4(\R)}^2}{N^{(1-\sigma)/2}} + \frac{\|\phi(t)\|_{L^\infty(\R)}^2}{N^{\delta/2}} +  N^{\frac{2(\sigma-1)+\delta}{2}} +\vep^{1/2}\|\phi(t)\|_{C^{1/2}(\R)}^2 \\
&\ph + \|\phi(t)\|_{C^{1/2}(\R)}^2\|\phi(t)\|_{H^1(\R)}^2 \beta_{N,\vep}(t)  +\paren*{1+\|\phi(t)\|_{C^{1/2}(\R)}^2}\|\nabla_1q_1(t)\Phi_N^\vep(t)\|_{L^2(\R^N)}^2,
\end{split}
\end{equation}
for every $t\in\R$, uniformly in $(\vep,\sigma,\delta)\in (0,1)^3$ and $N\in\N$.
\end{prop}
\begin{proof}
By time-reversal symmetry, it is enough to consider $t\geq 0$. Following the argument in \cite[Subsubsection 3.3.2, pg. 113]{KP2010}, we see that $\beta_{N,\varepsilon}$ is differentiable and its derivative $\dot{\beta}_{N,\varepsilon}$ is given by
\begin{equation}
\dot{\beta}_{N,\varepsilon} = i\kappa\ip{\Phi_N^\vep}{\comm{\frac{1}{N}\sum_{1\leq i<j\leq N} V_{\vep,ij} - \sum_{i=1}^N V_i^\phi}{\wh{n_N}}\Phi_N^\vep}_{L_{\ux_N}^2},
\end{equation}
where we have introduced the notation
\begin{equation}
\label{eq:V_i^phi_def}
V_{\vep,ij} \coloneqq V_\vep(X_i-X_j) \quad \text{and} \quad V_i^\phi \coloneqq |\phi(X_i)|^2.
\end{equation}
Using the symmetry of $\Phi_N^\vep$ and $\wh{n_N}$ with respect to exchange of particle labels and the decomposition $\id_N=(p_1+q_1)(p_2+q_2)$, then examining which terms cancel, we see that
\begin{align}
\dot{\beta}_{N,\vep} &= \frac{i\kappa}{2}\ip{\Phi_N^\vep}{\comm{(N-1)V_{\vep,12} - NV_{1}^\phi - NV_2^\phi}{\wh{n_N}}\Phi_N^\vep}_{L_{\ux_N}^2} \nn\\
&=\mathrm{Term}_1 + \mathrm{Term}_2 + \mathrm{Term}_3,
\end{align}
where
\begin{align}
\mathrm{Term}_1 &\coloneqq 2\Re{{i\kappa}\ip{\Phi_N^\vep}{p_1p_2\comm{(N-1)V_{\vep,12}-NV_1^\phi-NV_2^\phi}{\wh{n_N}}q_1p_2\Phi_N^\vep}_{L_{\ux_N}^2}}, \\
\mathrm{Term}_2 &\coloneqq 2\Re{i\kappa\ip{\Phi_N^\vep}{q_1p_2\comm{(N-1)V_{\vep,12}-NV_1^\phi-NV_2^\phi}{\wh{n_N}}q_1q_2\Phi_N^\vep}_{L_{\ux_N}^2}}, \\
\mathrm{Term}_3 &\coloneqq \Re{i\kappa\ip{\Phi_N^\vep}{p_1p_2\comm{(N-1)V_{\vep,12}-NV_1^\phi-NV_2^\phi}{\wh{n_N}}q_1q_2\Phi_N^\vep}_{L_{\ux_N}^2}}.
\end{align}
We proceed to estimate $\mathrm{Term}_1$, $\mathrm{Term}_2$, and $\mathrm{Term}_3$ individually. In the sequel, we drop the subscript $N$, as the number of particles is fixed. For convenience, we also introduce the notation
\begin{equation}
\label{eq:V_(ep,j)_def}
V_\vep^\phi(x) \coloneqq (V_\vep\ast|\phi|^2)(x) \enspace \text{and} \enspace V_{\vep,j}^\phi \coloneqq (V_\vep\ast|\phi|^2)(X_j), \qquad \forall j\in\{1,\ldots,N\}.
\end{equation}
Note that by Young's inequality and $\|V_\vep\|_{L^1}=1$, we have the operator norm estimate
\begin{equation}
\label{eq:V_ep_phi_Linf}
\|V_{\vep,j}^\phi\|_{L_{\ux_N}^2\rightarrow L_{\ux_N}^2} \leq \|\phi\|_{L_x^\infty}^2, \qquad \forall \varepsilon>0, \ j\in\{1,\ldots,N\}.
\end{equation}

\begin{description}[leftmargin=*]
\item[Estimate for $\mathrm{Term}_1$]
We first observe that since $q_1$ commutes with $V_2^\phi,\wh{n}$ and $p_1,q_1$ are orthogonal,
\begin{equation}
\label{eq:T1_p1q1_zero}
\ip{\Phi^\vep}{p_1p_2\comm{NV_2^\phi}{\wh{n}}q_1p_2\Phi^\vep}_{L_{\ux_N}^2} = \ip{\Phi^\vep}{\underbrace{p_1q_1}_{=0}p_2\comm{NV_2^\phi}{\wh{n}}p_2\Phi^\vep}_{L_{\ux_N}^2} =0.
\end{equation}
Since $p_2V_{\vep,12}p_2 = V_{\vep,1}^\phi p_2$, it follows that
\begin{align}
\left|\mathrm{Term}_1\right| &\lesssim \left|\ip{\Phi^\vep}{p_1p_2\comm{(N-1)V_{\vep,1}^\phi-NV_1^\phi}{\wh{n}}q_1p_2\Phi^\vep}_{L_{\ux_N}^2}\right| \nn\\
&=\left|\ip{\Phi^\vep}{p_1p_2 \paren*{(N-1)V_{\vep,1}^\phi-NV_1^\phi} (\wh{n}-\wh{(\tau_{-1}n)})q_1p_2\Phi^\vep}_{L_{\ux_N}^2}\right|,
\end{align}
where the ultimate equality follows from an application of \cref{lem:shift}. Define the function
\begin{equation}
\label{eq:mu_def}
\mu:\Z\rightarrow\R, \qquad \mu(k)\coloneqq N\paren*{n(k)-(\tau_{-1}n)(k)}, \qquad \forall k\in\Z,
\end{equation}
and observe that
\begin{equation}
\label{eq:mu_bnd}
\mu(k) = \frac{\sqrt{N}}{\sqrt{k}+1_{\geq 0}(k-1)\sqrt{k-1}}1_{\geq 0}(k) \leq n^{-1}(k), \qquad \forall k\in\Z.
\end{equation}
So by the triangle inequality,
\begin{align}
\left|\mathrm{Term}_1\right| &\lesssim \frac{1}{N}\left|\ip{\Phi^\vep}{p_1p_2V_{\vep,1}^\phi \wh{\mu}q_1p_2\Phi^\vep}_{L_{\ux_N}^2}\right| + \left|\ip{\Phi^\vep}{p_1p_2(V_{\vep,1}^\phi - V_1^\phi)\wh{\mu}q_1p_2\Phi^\vep}_{L_{\ux_N}^2}\right| \nn\\
&\leq \frac{1}{N}\|V_{\vep,1}^\phi\wh{\mu}q_1\Phi^\vep\|_{L_{\ux_N}^2} + \|(V_{\vep,1}^\phi-V_1^\phi)\wh{\mu}q_1\Phi^\vep\|_{L_{\ux_N}^2}, \label{eq:beta_ep_T1_op_app}
\end{align}
where the ultimate inequality follows from Cauchy-Schwarz and $\|\Phi^\vep\|_{L_{\ux_N}^2}=1$. By translation invariance of Lebesgue measure and $\int_{\R}dy V_\vep(y)=1$, for any $x\in\R$,
\begin{align}
\left|(V_\vep\ast |\phi|^2)(x) - |\phi(x)|^2\right| = \left|\int_{\R}dy V_\vep(y)\paren*{|\phi(x-y)|^2-|\phi(x)|^2}\right| \leq \int_{\R}dy V_\vep(y)|y|^{1/2} \| |\phi|^2\|_{\dot{C}_x^{1/2}} \leq \vep^{1/2}\|\phi\|_{C_x^{1/2}}^2
\end{align}
where the ultimate inequality follows from dilation invariance of Lebesgue measure. Hence,
\begin{equation}
\label{eq:op_norm_conv}
\|(V_\vep\ast|\phi|^2)-|\phi|^2\|_{L_x^\infty}\leq \vep^{1/2}\|\phi\|_{C_x^{1/2}}^2 \Longrightarrow \| V_{\vep,1}^\phi - V_1^\phi\|_{L_{\ux_N}^2\rightarrow L_{\ux_N}^2} \leq \vep^{1/2}\|\phi\|_{C_x^{1/2}}^2.
\end{equation}
Using the preceding operator norm estimate together with \eqref{eq:V_ep_phi_Linf}, we obtain that
\begin{align}
\eqref{eq:beta_ep_T1_op_app} &\leq \paren*{\frac{\|\phi\|_{L_x^\infty}^2}{N}+\vep^{1/2}\|\phi\|_{C_x^{1/2}}^2} \|\wh{\mu}q_1\Phi^\vep\|_{L_{\ux_N}^2} \lesssim \frac{\|\phi\|_{L_x^\infty}^2}{N}+\vep^{1/2}\|\phi\|_{C_x^{1/2}}^2,
\end{align}
where the ultimate inequality follows from the bound \eqref{eq:mu_bnd} for $\mu$ and an application of \cref{lem:mN_est}\ref{item:mN_q1} together with recalling that $\wh{n}^2 = \wh{m}$. Thus, we conclude that
\begin{equation}
\label{eq:beta_ep_T_1_fin}
\left|\mathrm{Term}_1\right| \lesssim \frac{\|\phi\|_{L_x^\infty}^2}{N}+\vep^{1/2}\|\phi\|_{C_x^{1/2}}^2.
\end{equation}
\item[Estimate for $\mathrm{Term}_2$]
Arguing similarly as in \eqref{eq:T1_p1q1_zero}, we see that
\begin{equation}
\ip{\Phi^\vep}{q_1p_2\comm{V_{1}^\phi}{\wh{n}}q_1q_2\Phi^\vep}_{L_{\ux_N}^2} = 0.
\end{equation}
Therefore,
\begin{align}
2\left|\mathrm{Term}_2\right| &= \left|\ip{\Phi^\vep}{q_1p_2\comm{(N-1)V_{\vep,12}-NV_{2}^\phi}{\wh{n}}q_1q_2\Phi^\vep}_{L_{\ux_N}^2}\right| \nn\\
&= \left|\ip{\Phi^\vep}{q_1p_2\paren*{\frac{(N-1)}{N}V_{\vep,12}-V_{2}^\phi}\wh{\mu}q_1q_2\Phi^\vep}_{L_{\ux_N}^2}\right| \nn\\
&\leq \underbrace{\left|\ip{\Phi^\vep}{q_1p_2V_{\vep,12}\wh{\mu}q_1q_2\Phi^\vep}_{L_{\ux_N}^2}\right|}_{\eqqcolon\mathrm{Term}_{2,1}} + \underbrace{\left|\ip{\Phi^\vep}{q_1p_2V_{2}^\phi\wh{\mu}q_1q_2\Phi^\vep}_{L_{\ux_N}^2}\right|}_{\eqqcolon\mathrm{Term}_{2,2}},
\end{align}
where to obtain the penultimate equality have used \cref{lem:shift} and introduced the notation $\mu$ from \eqref{eq:mu_def} and to obtain the ultimate equality we have used the triangle inequality.

We first consider $\mathrm{Term}_{2,2}$. By Cauchy-Schwarz together with the estimate \eqref{eq:V_ep_phi_Linf},
\begin{equation}
\mathrm{Term}_{2,2} \leq \|q_1\Phi^\vep\|_{L_{\ux_N}^2} \|p_2 V_{2}^\phi\wh{\mu}q_1q_2\Phi^\vep\|_{L_{\ux_N}^2} \leq \|q_1\Phi^\vep\|_{L_{\ux_N}^2} \|\phi\|_{L_x^\infty}^2 \|\wh{\mu}q_1q_2\Phi^\vep\|_{L_{\ux_N}^2}.
\end{equation}
By \cref{rem:alpha_N} and \cref{lem:mN_est}\ref{item:mN_q1q2}, respectively, together with the $\mu$ bound \eqref{eq:mu_bnd}, we have that
\begin{equation}
\label{eq:mu_q1q2}
\|q_1\Phi^\vep\|_{L_{\ux_N}^2} \leq \sqrt{\alpha_\vep} \leq \sqrt{\beta_\vep} \quad \text{and} \quad \|\wh{\mu}q_1q_2\Phi^\vep\|_{L_{\ux_N}^2}\lesssim \sqrt{\beta_\vep}.
\end{equation}
Therefore,
\begin{equation}
\label{eq:beta_ep_T_22_fin}
\mathrm{Term}_{2,2} \lesssim \|\phi\|_{L_x^\infty}^2\beta_\vep.
\end{equation}

We now consider $\mathrm{Term}_{2,1}$. It follows from the identities \eqref{eq:intro_sgn} and $\delta \ast V_\vep=V_\vep$ that
\begin{equation}
V_\vep =  \frac{1}{2}\nabla(\sgn\ast V_\vep).
\end{equation}
We introduce the notation $X_{\vep,12}\coloneqq \frac{1}{2}(\sgn\ast V_\vep)(X_1-X_2)$. Note, $\|V_\vep\|_{L^1}=\|\sgn\|_{L^\infty}=1$, so that
\begin{equation}
\label{eq:X_ep_opnorm}
\|X_{\vep,12}\|_{L_{\ux_N}^2\rightarrow L_{\ux_N}^2} \leq \frac{1}{2}.
\end{equation}
Integrating by parts and applying the product rule and triangle inequality,
\begin{align}
\mathrm{Term}_{2,1} \leq \left|\ip{\nabla_1 q_1p_2\Phi^\vep}{X_{\vep,12}\wh{\mu}q_1q_2\Phi^\vep}_{L_{\ux_N}^2}\right| + \left|\ip{\Phi^\vep}{q_1p_2X_{\vep,12}\nabla_1\wh{\mu}q_1q_2\Phi^\vep}_{L_{\ux_N}^2}\right| \eqqcolon \mathrm{Term}_{2,1,1} +\mathrm{Term}_{2,1,2}.
\end{align}
By Cauchy-Schwarz and the estimate \eqref{eq:X_ep_opnorm},
\begin{align}
\mathrm{Term}_{2,1,1} &\leq \|\nabla_1 q_1p_2\Phi^\vep\|_{L_{\ux_N}^2} \|\wh{\mu}q_1q_2\Phi^\vep\|_{L_{\ux_N}^2},
\end{align}
so by application of the second estimate of \eqref{eq:mu_q1q2} and $\|p_2\|_{L_{\ux_N}^2\rightarrow L_{\ux_N}^2}=1$,
\begin{equation}
\label{eq:beta_ep_T_211_fin}
\mathrm{Term}_{2,1,1} \lesssim \|\nabla_1q_1\Phi^\vep\|_{L_{\ux_N}^2}\sqrt{\beta_\vep}.
\end{equation}
Next, we write $\id=p_1+q_1$ and use the triangle inequality to obtain
\begin{equation}
\mathrm{Term}_{2,1,2} \leq \left|\ip{p_2q_1\Phi^\vep}{X_{\vep,12}p_1\nabla_1\wh{\mu}q_1q_2\Phi^\vep}_{L_{\ux_N}^2}\right| + \left|\ip{p_2q_1\Phi^\vep}{X_{\vep,12}q_1\nabla_1\wh{\mu}q_1q_2\Phi^\vep}_{L_{\ux_N}^2}\right|.
\end{equation}
By \cref{lem:shift}, we have the operator identity
\begin{equation}
p_1\nabla_1\wh{\mu}q_1 = p_1\wh{(\tau_1\mu)}\nabla_1 q_1.
\end{equation}
Hence,
\begin{align}
\left|\ip{p_2q_1\Phi^\vep}{X_{\vep,12}p_1\nabla_1\wh{\mu}q_1q_2\Phi^\vep}_{L_{\ux_N}^2}\right| &\leq \|X_{\vep,12}p_2q_1\Phi^\vep\|_{L_{\ux_N}^2} \|p_1\wh{(\tau_1\mu)}\nabla_1q_1q_2\Phi^\vep\|_{L_{\ux_N}^2} \nn\\
&\leq \|q_1\Phi^\vep\|_{L_{\ux_N}^2} \|\wh{(\tau_1\mu)}\nabla_1q_1q_2\Phi^\vep\|_{L_{\ux_N}^2}.
\end{align}
By \cref{rem:alpha_N}, $\|q_1\Phi^\vep\|_{L_{\ux_N}^2}\leq \sqrt{\beta_\vep}$. Now using the $\mu$ bound \eqref{eq:mu_bnd}, we have that
\begin{equation}
(\tau_1\mu)(k) \lesssim n^{-1}(k+1) \lesssim n^{-1}(k), \qquad \forall k\in\Z.
\end{equation}
Combining this estimate with the symmetry of $\Phi^\vep$ under permutation of particle labels, we find that
\begin{align}
\|\wh{(\tau_1\mu)}\nabla_1q_1q_2\Phi^\vep\|_{L_{\ux_N}^2} &\lesssim \sqrt{\ip{\nabla_1q_1\Phi^\vep}{\wh{n}^{-2}\nabla_1 q_1q_2\Phi^\vep}_{L_{\ux_N}^2}} \nn\\
&=\sqrt{\frac{1}{N-1}\sum_{i=2}^N \ip{\nabla_1q_1\Phi^\vep}{q_i\wh{n}^{-2}\nabla_1q_1\Phi^\vep}_{L_{\ux_N}^2}}.
\end{align}
Since the projector $q_1$ commutes with $\wh{n}^{-2}$ and $\wh{n}^{-2}\geq 0$, we have that
\begin{equation}
\ip{\nabla_1 q_1\Phi^\vep}{q_1\wh{n}^{-2}\nabla_1 q_1}_{L_{\ux_N}^2} = \ip{q_1\nabla_1 q_1\Phi^\vep}{\wh{n}^{-2}q_1\nabla_1 q_1\Phi^\vep}_{L_{\ux_N}^2}\geq 0,
\end{equation}
so that by \cref{rem:alpha_N} and the identity $n^2=m$,
\begin{align}
\label{eq:sym_q2}
\sqrt{\frac{1}{N-1}\sum_{i=2}^N \ip{\nabla_1q_1\Phi^\vep}{q_i\wh{n}^{-2}\nabla_1q_1\Phi^\vep}_{L_{\ux_N}^2}} &\lesssim \sqrt{\frac{1}{N}\sum_{i=1}^N\ip{\nabla_1 q_1\Phi^\vep}{q_i\wh{n}^{-2}\nabla_1q_1\Phi^\vep}_{L_{\ux_N}^2}} \nn\\
&=\sqrt{\ip{\nabla_1 q_1\Phi^\vep}{\wh{n}^{-2}\wh{n}^{2}\nabla_1q_1\Phi^\vep}_{L_{\ux_N}^2}} \nn\\
&= \|\nabla_1q_1\Phi^\vep\|_{L_{\ux_N}^2}.
\end{align}
After a little bookkeeping, we find that
\begin{equation}
\left|\ip{p_2q_1\Phi^\vep}{X_{\vep,12}p_1\nabla_1\wh{\mu}q_1q_2\Phi^\vep}_{L_{\ux_N}^2}\right| \lesssim \sqrt{\beta_\vep}\|\nabla_1q_1\Phi^\vep\|_{L_{\ux_N}^2}.
\end{equation}
Again by \cref{lem:shift}, we have the operator identity
\begin{equation}
q_1\nabla_1\wh{\mu}q_1 = q_1\wh{\mu}\nabla_1q_1,
\end{equation}
and proceeding similarly as immediately above, we find that
\begin{equation}
\left|\ip{p_2q_1\Phi^\vep}{X_{\vep,12}q_1\nabla_1\wh{\mu}q_1q_2\Phi^\vep}_{L_{\ux_N}^2}\right| \lesssim \sqrt{\beta_\vep}\|\nabla_1q_1\Phi^\vep\|_{L_{\ux_N}^2},
\end{equation}
and therefore
\begin{equation}
\label{eq:beta_ep_T_212_fin}
\mathrm{Term}_{2,1,2} \lesssim \|\nabla_1 q_1\Phi^\vep\|_{L_{\ux_N}^2}\sqrt{\beta_\vep}.
\end{equation}
Together the estimate \eqref{eq:beta_ep_T_211_fin} for $\mathrm{Term}_{2,1,1}$, we obtain that
\begin{equation}
\label{eq:beta_ep_T_21_fin}
\mathrm{Term}_{2,1} \lesssim \|\nabla_1 q_1\Phi^\vep\|_{L_{\ux_N}^2}\sqrt{\beta_\vep}.
\end{equation}

Collecting the estimates \eqref{eq:beta_ep_T_21_fin} for $\mathrm{Term}_{2,1}$ and \eqref{eq:beta_ep_T_22_fin} for $\mathrm{Term}_{2,2}$, we conclude that
\begin{equation}
\label{eq:beta_ep_T_2_fin}
\mathrm{Term}_{2} \lesssim \|\phi\|_{L_x^\infty}^2\beta_\vep + \|\nabla_1 q_1\Phi^\vep\|_{L_{\ux_N}^2}\sqrt{\beta_\vep}.
\end{equation}

\item[Estimate for $\mathrm{Term}_3$]
We now consider $\mathrm{Term}_3$, which is the most difficult portion of the analysis. We first note that by arguing similarly as in \eqref{eq:T1_p1q1_zero}, we see that
\begin{equation}
p_1p_2\comm{V_{1}^\phi}{\wh{n}}q_1q_2 = 0 = p_1p_2\comm{V_{2}^\phi}{\wh{n}}q_1q_2,
\end{equation}
where the reader will recall the notation $V_{j}^\phi$ introduced in \eqref{eq:V_i^phi_def}. Therefore,
\begin{align}
\left|\mathrm{Term}_3\right| &\lesssim \left|\ip{\Phi^\vep}{p_1p_2\comm{(N-1)V_{\vep,12}}{\wh{n}}q_1q_2\Phi^\vep}_{L_{\ux_N}^2}\right| \nn\\
&=\frac{N-1}{N}\left|\ip{\Phi^\vep}{p_1p_2 NV_{\vep,12}\paren*{\wh{n}-\wh{(\tau_{-2}n}}q_1q_2\Phi^{\vep}}_{L_{\ux_N}^2}\right|,
\end{align}
where the ultimate equality follows from unpacking the commutator and applying \cref{lem:shift}. Analogously to the function $\mu$ defined in \eqref{eq:mu_def}, we define the function
\begin{equation}
\label{eq:nu_def}
\nu:\Z\rightarrow\R, \qquad \nu(k) \coloneqq N\paren*{n(k)-(\tau_{-2}n)(k)}, \qquad \forall k\in\Z.
\end{equation}
It is a straightforward computation from the definition of $n$ in \cref{def:mn_N} that
\begin{equation}
\nu(k) = \frac{2\sqrt{N}}{\sqrt{k}+1_{\geq 2}(k)\sqrt{k-2}}1_{\geq 0}(k), \qquad \forall k\in \Z,
\end{equation}
which implies that
\begin{equation}
\label{eq:nu_bnd}
\nu(k) \lesssim n^{-1}(k), \qquad \forall k\in\Z.
\end{equation}

We now introduce an approximation of the pair potential $V_\ep$ as follows. Define $V_\sigma(x) \coloneqq N^\sigma \tl{V}(N^\sigma x)$, where $\sigma\in(0,1)$ is a parameter to be specified momentarily and $\tl{V}$ is a standard mollifier. We convolve $V_\vep$ with $V_\sigma$ to define
\begin{equation}
\label{eq:V_ep_beta_def}
V_{\vep,\sigma} \coloneqq V_\vep \ast V_\sigma \quad \text{and} \quad V_{\vep,\sigma,ij}\coloneqq V_{\vep,\sigma}(X_i-X_j), \enspace \forall 1\leq i<j\leq N.
\end{equation}
By the triangle inequality,
\begin{equation}
\label{eq:T_3_12_def}
\left|\ip{\Phi^\vep}{p_1p_2V_{\vep,12}\wh{\nu}q_1q_2\Phi^\vep}_{L_{\ux_N}^2}\right| \leq \underbrace{\left|\ip{\Phi^\vep}{p_1p_2(V_{\vep,12}-V_{\vep,\sigma,12})\wh{\nu}q_1q_2\Phi^\vep}_{L_{\ux_N}^2}\right|}_{\eqqcolon \mathrm{Term}_{3,1}} + \underbrace{\left|\ip{\Phi^\vep}{p_1p_2V_{\vep,\sigma,12}\wh{\nu}q_1q_2\Phi^\vep}_{L_{\ux_N}^2}\right|}_{\eqqcolon\mathrm{Term}_{3,2}}.
\end{equation}
Observe that by moving $p_1p_2$ over to the first entry of the inner product, writing out the convolution implicit in $V_{\vep,\sigma,12}$, and using the Fubini-Tonelli theorem, we have that
\begin{align}
&\ip{\Phi^\vep}{p_1p_2 V_{\vep,\sigma,12} \wh{\nu}q_1q_2\Phi^\vep}_{L_{\ux_N}^2} \nn\\
&=\int_{\R}dy V_\sigma(y)\int_{\R^2}d\ux_{1;2}V_\vep(x_1-x_2-y)\int_{\R^{N-2}}d\ux_{3;N} \paren*{\ol{(p_1p_2\Phi^\vep)}(\wh{\nu}q_1q_2\Phi^\vep)}(x_1,x_2,\ux_{3;N}) \nn\\
&=\int_{\R}dy V_\sigma(y)\int_{\R^2}d\ux_{1;2}V_\vep(x_1-x_2-y)\int_{\R^{N-2}}d\ux_{3;N} \left(\paren*{\ol{(p_1p_2\Phi^\vep)}(\wh{\nu}q_1q_2\Phi^\vep)}(x_1,x_2,\ux_{3;N})\right. \nn\\
&\hspace{95mm}- \left.\paren*{\ol{(p_1p_2\Phi^\vep)}(\wh{\nu}q_1q_2\Phi^\vep)}(x_1,x_2+y,\ux_{3;N})\right) \nn\\
&\ph +\int_{\R}dyV_\sigma(y)\int_{\R^2}d\ux_{1;2}V_\vep(x_1-x_2-y)\int_{\R^{N-2}}d\ux_{3;N} \paren*{\ol{(p_1p_2\Phi^\vep)} (\wh{\nu}q_1q_2\Phi^\vep)}(x_1,x_2+y,\ux_{3;N}).
\end{align}
By translation invariance of Lebesgue measure applied in the $x_2$-coordinate, we have that for any $y\in\R$,
\begin{align}
&\int_{\R^2}d\ux_{1;2}V_\vep(x_1-x_2-y)\int_{\R^{N-2}}d\ux_{3;N} \paren*{\ol{(p_1p_2\Phi^\vep)} (\wh{\nu}q_1q_2\Phi^\vep)}(x_1,x_2+y,\ux_{3;N}) \nn\\
&=\int_{\R^2}d\ux_{1;2}V_{\vep}(x_1-x_2)\int_{\R^{N-2}}d\ux_{3;N} \paren*{\ol{(p_1p_2\Phi^\vep)} (\wh{\nu}q_1q_2\Phi^\vep)}(x_1,x_2,\ux_{3;N}) \nn\\
&=\ip{\Phi^\vep}{p_1p_2V_{\vep,12}\wh{\nu}q_1q_2\Phi^\vep}_{L_{\ux_N}^2},
\end{align}
where the ultimate equality follows from using the Fubini-Tonelli theorem and the self-adjointness of $p_1p_2$. Since $\int_{\R}dyV_\sigma(y)=1$, we conclude that
\begin{equation}
\begin{split}
&\int_{\R}dyV_\sigma(y)\int_{\R^2}d\ux_{1;2}V_\vep(x_1-x_2-y)\int_{\R^{N-2}}d\ux_{3;N} \paren*{\ol{(p_1p_2\Phi^\vep)} (\wh{\nu}q_1q_2\Phi^\vep)}(x_1,x_2+y,\ux_{3;N})\\
&=\ip{\Phi^\vep}{p_1p_2V_{\vep,12}\wh{\nu}q_1q_2\Phi^\vep}_{L_{\ux_N}^2}.
\end{split}
\end{equation}
Next, we have by definition of the H\"older norm in the $x_2$-coordinate that
\begin{align}
&\sup_{x_2\in\R}\left|\paren*{\ol{(p_1p_2\Phi^\vep)}(\wh{\nu}q_1q_2\Phi^\vep)}(x_1,x_2,\ux_{3;N})-\paren*{\ol{(p_1p_2\Phi^\vep)}(\wh{\nu}q_1q_2\Phi^\vep)}(x_1,x_2+y,\ux_{3;N})\right| \nn\\
&\leq \|(p_1p_2\Phi^\vep)(x_1,\cdot,\ux_{3;N})\|_{C_{x_2}^{1/2}} \|(\wh{\nu}q_1q_2\Phi^\vep)(x_1,\cdot,\ux_{3;N})\|_{C_{x_2}^{1/2}} |y|^{1/2},
\end{align}
for every $y\in\R$ and almost every $(x_1,\ux_{3;N})\in\R^{N-1}$. So by the Fubini-Tonelli theorem, followed by using the translation and dilation invariance of Lebesgue measure and then Cauchy-Schwarz, we find that
\begin{align}
&\int_{\R}dy V_\sigma(y)\int_{\R^2}d\ux_{1;2}V_\vep(x_1-x_2-y)\int_{\R^{N-2}}d\ux_{3;N} \left|\paren*{\ol{(p_1p_2\Phi^\vep)}(\wh{\nu}q_1q_2\Phi^\vep)}(x_1,x_2,\ux_{3;N})\right. \nn\\
&\hspace{90mm}- \left.\paren*{\ol{(p_1p_2\Phi^\vep)}(\wh{\nu}q_1q_2\Phi^\vep)}(x_1,x_2+y,\ux_{3;N})\right| \nn\\
&\leq \int_{\R^{N-1}}dx_1d\ux_{3;N}\Bigg(\|(p_1p_2\Phi^\vep)(x_1,\cdot,\ux_{3;N})\|_{C_{x_2}^{1/2}} \|(\wh{\nu}q_1q_2\Phi^\vep)(x_1,\cdot,\ux_{3;N})\|_{C_{x_2}^{1/2}} \nn\\
&\hspace{50mm} \times\underbrace{\paren*{\int_{\R}dy|y|^{1/2} V_\sigma(y)\int_{\R}dx_2  V_\vep(x_1-x_2-y)}}_{\lesssim N^{-\sigma/2}}\Bigg) \nn\\
&\lesssim N^{-\sigma/2}\|p_1p_2\Phi^\vep\|_{L_{\ux_{2;N}}^2C_{x_1}^{1/2}} \|\wh{\nu}q_1q_2\Phi^\vep\|_{L_{\ux_{2;N}}^2C_{x_1}^{1/2}},
\end{align}
where in the ultimate inequality we use the symmetry of $\Phi_\vep$ to swap $x_1$ and $x_2$ in order to ease the burden of notation. By Fubini-Tonelli, Cauchy-Schwarz, and the normalization $\|\phi\|_{L_x^2}=1$, we have the estimate
\begin{equation}
\label{eq:p1p2_hold}
\|p_1p_2\Phi^\vep\|_{L_{\ux_{2;N}}^2C_{x_1}^{1/2}} \leq \|\phi\|_{C_x^{1/2}}\|p_2\Phi^\vep\|_{L_{\ux_N}^2} \leq \|\phi\|_{C_x^{1/2}},
\end{equation}
where the ultimate inequality follows from the normalization $\|\Phi^{\vep}\|_{L_{\ux_N}^2}=1$. By \cref{lem:H1_hold} and the $H^{1/2+}\subset L^\infty$ Sobolev embedding,
\begin{equation}
\|\wh{\nu}q_1q_2\Phi^\vep\|_{L_{\ux_{2;N}}^2C_{x_1}^{1/2}} \lesssim \|\wh{\nu}q_1q_2\Phi^\vep\|_{L_{\ux_{2;N}}^2H_{x_1}^1} \lesssim \|\wh{\nu}q_1q_2\Phi^\vep\|_{L_{\ux_N}^2} + \|\nabla_1\wh{\nu}q_1q_2\Phi^\vep\|_{L_{\ux_N}^2},
\end{equation}
where the ultimate inequality follows from splitting the $H_{x}^1$ norm and Fubini-Tonelli. Using the $\nu$ estimate \eqref{eq:nu_bnd}, \cref{lem:mN_est}\ref{item:mN_q1q2}, and the identity $\wh{m}=\wh{n}^2$, we see that
\begin{align}
\label{eq:nu_q1q2}
\|\wh{\nu}q_1q_2\Phi^\vep\|_{L_{\ux_N}^2} \lesssim \sqrt{\ip{\Phi^\vep}{\wh{n}^{-2}\wh{m}^{2}\Phi^\vep}_{L_{\ux_N}^2}} =\sqrt{\ip{\Phi^\vep}{\wh{m}\Phi^\vep}_{L_{\ux_N}^2}} =\sqrt{\alpha_\vep} \leq \sqrt{\beta_\vep}.
\end{align}
Next, inserting the decomposition $\nabla_1=p_1\nabla_1+q_1\nabla_1$ and applying the triangle inequality,
\begin{equation}
\label{eq:grad_nu_q1q2}
\|\nabla_1\wh{\nu}q_1q_2\Phi^\vep\|_{L_{\ux_N}^2} \leq \|p_1\nabla_1\wh{\nu}q_1q_2\Phi^\vep\|_{L_{\ux_N}^2} + \|q_1\nabla_1\wh{\nu}q_1q_2\Phi^\vep\|_{L_{\ux_N}^2}.
\end{equation}
Since $p_1\nabla_1 = -(\ket*{\phi}\bra*{\nabla\phi})_1$,
\begin{equation}
\|p_1\nabla_1\wh{\nu}q_1q_2\Phi^\vep\|_{L_{\ux_N}^2} \leq \|\nabla\phi\|_{L_x^2} \|\wh{\nu}q_1q_2\Phi^\vep\|_{L_{\ux_N}^2} \lesssim \|\nabla\phi\|_{L_x^2}\sqrt{\beta_\vep},
\end{equation}
where the ultimate inequality follows from the estimate \eqref{eq:nu_q1q2}. By \cref{lem:shift} followed by using the $\nu$ estimate \eqref{eq:nu_bnd},
\begin{equation}
\|q_1\nabla_1\wh{\nu}q_1q_2\Phi^\vep\|_{L_{\ux_N}^2} = \|q_1\wh{\nu}\nabla_1q_1q_2\Phi^\vep\|_{L_{\ux_N}^2} \lesssim \sqrt{\ip{\nabla_1q_1\Phi^\vep}{q_2\wh{n}^{-2}\nabla_1q_1\Phi^\vep}_{L_{\ux_N}^2}},
\end{equation}
and arguing as for the estimate \eqref{eq:sym_q2}, we find that the right-hand side is $\lesssim \|\nabla_1q_1\Phi^\vep\|_{L_{\ux_N}^2}$. Therefore,
\begin{equation}
\label{eq:nu_q1q2_hold}
\|\wh{\nu}q_1q_2\Phi^\vep\|_{L_{\ux_{2;N}}^2C_{x_1}^{1/2}} \lesssim \paren*{1+\|\nabla\phi\|_{L_x^2}}\sqrt{\beta_\vep} + \|\nabla_1 q_1\Phi^\vep\|_{L_{\ux_N}^2} \lesssim \|\phi\|_{H_x^1}\sqrt{\beta_\vep} + \|\nabla_1 q_1\Phi^\vep\|_{L_{\ux_N}^2}.
\end{equation}
Collecting the estimates \eqref{eq:p1p2_hold}, \eqref{eq:nu_q1q2_hold} and applying Young's inequality for products, we see that
\begin{align}
N^{-\sigma/2}\|p_1p_2\Phi^\vep\|_{L_{\ux_{2;N}}^2C_{x_1}^{1/2}} \|\wh{\nu}q_1q_2\Phi^\vep\|_{L_{\ux_{2;N}}^2C_{x_1}^{1/2}} &\lesssim N^{-\sigma} + \|\phi\|_{C_x^{1/2}}^2\|\phi\|_{H_x^1}^2\beta_\vep + \|\phi\|_{C_x^{1/2}}^2\|\nabla_1q_1\Phi^\vep\|_{L_{\ux_N}^2}^2.
\end{align}

After a little bookkeeping, we conclude that
\begin{equation}
\label{eq:beta_ep_T_31_fin}
\left|\mathrm{Term}_{3,1}\right| \lesssim N^{-\sigma} + \|\phi\|_{C_x^{1/2}}^2\|\phi\|_{H_x^1}^2\beta_\vep + \|\phi\|_{C_x^{1/2}}^2\|\nabla_1q_1\Phi^\vep\|_{L_{\ux_N}^2}^2,
\end{equation}
leaving us with $\mathrm{Term}_{3,2}$.

For $\mathrm{Term}_{3,2}$, we borrow an idea from \cite{KP2010} and introduce a partition of unity as follows. Let $\chi^{(1)},\chi^{(2)}:\Z\rightarrow [0,\infty)$ be the two functions respectively defined by
\begin{equation}
\label{eq:chi(1)(2)_def}
\chi^{(1)}(k) \coloneqq 1_{\leq N^{1-\delta}}(k), \quad \chi^{(2)}(k) \coloneqq 1-\chi^{(1)}(k) = 1_{>N^{1-\delta}}(k), \qquad \forall k\in\Z.
\end{equation}
where $\delta\in(0,1)$ will be optimized at the end. Trivially, we have that $\chi^{(j)}\in\{0,1\}^{\Z}$, so that $(\chi^{(j)}(k))^2 = \chi^{(j)}(k)$, and $\chi^{(1)}(k)+\chi^{(2)}(k)=1$. We insert this decomposition into the expression for $\mathrm{Term}_{3,2}$ and use the triangle inequality to obtain
\begin{equation}
\label{eq:T_32j_def}
\left|\mathrm{Term}_{3,2}\right| \leq \underbrace{\left|\ip{\Phi^\vep}{p_1p_2 V_{\vep,\sigma,12}\wh{\nu}\wh{\chi^{(1)}}q_1q_2\Phi^\vep}_{L_{\ux_N}^2}\right|}_{\eqqcolon \mathrm{Term}_{3,2,1}} + \underbrace{\left|\ip{\Phi^\vep}{p_1p_2V_{\vep,\sigma,12}\wh{\nu}\wh{\chi^{(2)}}q_1q_2\Phi^\vep}_{L_{\ux_N}^2}\right|}_{\eqqcolon \mathrm{Term}_{3,2,2}}.
\end{equation}
We consider $\mathrm{Term}_{3,2,1}$ and $\mathrm{Term}_{3,2,2}$ separately.

For $\mathrm{Term}_{3,2,1}$, we want to use the fact that the operator norm of $p_1p_2 V_{\vep,\sigma,12}q_1q_2$ is much smaller on the bosonic subspace $L_{sym}^2(\R^N)$ than on the full space $L^2(\R^N)$. Accordingly, we symmetrize the expression $p_2 V_{\vep,\sigma,12}q_2$ to write
\begin{align}
\mathrm{Term}_{3,2,1} &= \frac{1}{N-1}\left|\ip{\Phi^\vep}{\sum_{i=2}^{N} p_1p_i V_{\vep,\sigma,1i}q_iq_1 \wh{\chi^{(1)}}\wh{\nu}q_1\Phi^\vep}_{L_{\ux_N}^2}\right| \nn\\
&\leq \frac{1}{N-1}\|\sum_{i=2}^N \wh{\chi^{(1)}}q_iq_1 V_{\vep,\sigma,1i}p_ip_1\Phi^\vep\|_{L_{\ux_N}^2} \|\wh{\nu}q_1\Phi^\vep\|_{L_{\ux_N}^2}.
\end{align}
where the ultimate line follows from Cauchy-Schwarz. We claim that $\|\wh{\nu}q_1\Phi^\vep\|_{L_{\ux_N}^2}\lesssim 1$. Indeed, by the $\nu$ bound \eqref{eq:nu_bnd} and \cref{lem:mN_est}\ref{item:mN_q1},
\begin{equation}
\label{eq:nu_q1_L2}
\|\wh{\nu}q_1\Phi^\vep\|_{L_{\ux_N}^2} = \sqrt{\ip{\Phi^\vep}{\wh{\nu}^2 q_1\Phi^\vep}_{L_{\ux_N}^2}} \lesssim \sqrt{\ip{\Phi^\vep}{\wh{n}^{-2}\wh{m}\Phi^\vep}_{L_{\ux_N}^2}} = 1,
\end{equation}
since $\wh{n}^2=\wh{m}$ and $\|\Phi^\vep\|_{L_{\ux_N}^2}=1$. Now expanding the $L_{\ux_N}^2$ norm and using that $\wh{\chi^{(1)}}^2 = \wh{\chi^{(1)}}$, we see that
\begin{align}
\|\sum_{i=2}^N \wh{\chi^{(1)}}q_iq_1 V_{\vep,\sigma,1i}p_ip_1\Phi^\vep\|_{L_{\ux_N}^2} &= \sqrt{\sum_{i,j=2}^N \ip{\Phi^\vep}{p_1p_i V_{\vep,\sigma,1i}q_1q_i\wh{\chi^{(1)}}q_1q_j V_{\vep,\sigma,1j}p_jp_1\Phi^\vep}_{L_{\ux_N}^2}} \nn\\
&\leq\underbrace{\sqrt{\sum_{i=2}^N \ip{\Phi^\vep}{p_1p_iV_{\vep,\sigma,1i}q_1q_i\wh{\chi^{(1)}}q_1q_i V_{\vep,\sigma,1i}p_ip_1\Phi^\vep}_{L_{\ux_N}^2}}}_{\eqqcolon \sqrt{B}} \nn\\
&\ph + \underbrace{\sqrt{\sum_{2\leq i\neq j\leq N} \ip{\Phi^\vep}{p_1p_iV_{\vep,\sigma,1i}q_1q_i\wh{\chi^{(1)}}q_1q_j V_{\vep,\sigma,1j}p_jp_1\Phi^\vep}_{L_{\ux_N}^2}}}_{\eqqcolon \sqrt{A}}, \label{eq:chi(1)_AB}
\end{align}
where the ultimate inequality follows from the embedding $\ell^{1/2}\subset \ell^1$. Therefore,
\begin{equation}
\label{eq:T_321_AB}
\mathrm{Term}_{3,2,1} \lesssim \frac{1}{N-1}\paren*{\sqrt{B}+\sqrt{A}}.
\end{equation}

We first consider $B$, which is the easy term. Since $\|q_1q_i\wh{\chi^{(1)}}q_1q_i\|_{L_{\ux_N}^2\rightarrow L_{\ux_N}^2} \leq 1$,
\begin{equation}
B \leq \sum_{i=2}^N \|V_{\vep,\sigma,1i}p_1p_i\Phi^\vep\|_{L_{\ux_N}^2}^2 = \sum_{i=2}^N \ip{\Phi^\vep}{p_1p_i V_{\vep,\sigma,1i}^2p_1p_i\Phi^\vep}_{L_{\ux_N}^2}.
\end{equation}
Now by examination of the integral kernel of $p_1p_i V_{\vep,\sigma,1i}^2p_1p_i$,
\begin{align}
p_1p_i V_{\vep,\sigma,1i}^2 p_1p_i &= \paren*{\int_{\R^2}dy_1dy_i V_{\vep,\sigma}^2(y_1-y_i) |\phi(y_1)|^2 |\phi(y_i)|^2}p_1p_i = \||\phi|^2 (V_{\vep,\sigma}^2\ast|\phi|^2)\|_{L_x^1}p_1p_i,
\end{align}
and by Cauchy-Schwarz followed by Young's inequality,
\begin{align}
\||\phi|^2 (V_{\vep,\sigma}^2\ast|\phi|^2)\|_{L_x^1} \leq \|\phi\|_{L_x^4}^2 \|V_{\vep,\sigma}^2\ast |\phi|^2\|_{L_x^2} \leq \underbrace{\|V_{\vep,\sigma}\|_{L^2}^2}_{\leq N^{\sigma}} \|\phi\|_{L_x^4}^4.
\end{align}
It then follows from $\|\Phi^\vep\|_{L_{\ux_N}^2}=1$ that
\begin{equation}
\label{eq:B_chi1_fin}
B\leq (N-1)N^{\sigma}\|\phi\|_{L_x^4}^4.
\end{equation}

We proceed to consider $A$. We first make a further decomposition of $A$ by using that $(\chi^{(1)})^2 = \chi^{(1)}$ and then applying \cref{lem:shift} in order to obtain
\begin{align}
A &= \sum_{2\leq i\neq j\leq N} \ip{\Phi^\vep}{p_1p_i V_{\vep,\sigma,1i}q_1q_i \wh{\chi^{(1)}}\wh{\chi^{(1)}} q_jq_1 V_{\vep,\sigma,1j}p_jp_1\Phi^\vep }_{L_{\ux_N}^2} \nn\\
&=\underbrace{\sum_{2\leq i\neq j\leq N} \ip{\Phi^\vep}{p_1p_iq_j \wh{(\tau_2\chi^{(1)})}V_{\vep,\sigma,1,i}V_{\vep,\sigma,1j}\wh{(\tau_2\chi^{(1)})}q_ip_jp_1\Phi^\vep}_{L_{\ux_N}^2}}_{\eqqcolon A_1} \nn\\
&\ph \underbrace{-\sum_{2\leq i\neq j\leq N} \ip{\Phi^\vep}{p_1p_iq_j\wh{(\tau_2\chi^{(1)})}V_{\vep,\sigma,1i}p_1 V_{\vep,\sigma,1j} \wh{(\tau_2\chi^{(1)})}q_ip_jp_1\Phi^\vep}_{L_{\ux_N}^2}}_{\eqqcolon A_2}, \label{eq:chi(1)_A_split}
\end{align}
where the ultimate equality follows from writing $q_1=\id -p_1$.

For $A_1$, we have by the triangle inequality and self-adjointness of $\wh{(\tau_2\chi^{(1)})}q_j$ that
\begin{align}
|A_1| &\leq \sum_{2\leq i\neq j\leq N}\left|\ip{\wh{(\tau_2\chi^{(1)})}q_j\Phi^\vep}{p_1p_iV_{\vep,\sigma,1i}V_{\vep,\sigma,1j}p_jp_1\wh{(\tau_2\chi^{(1)})}q_i\Phi^\vep}_{L_{\ux_N}^2}\right|.
\label{eq:A1_chi(1)_sym}
\end{align}
Using that $V_{\vep,\sigma}\geq 0$ and commutativity of point-wise multiplication operators, we can write
\begin{equation}
V_{\vep,\sigma,1i}V_{\vep,\sigma,1j} = (V_{\vep,\sigma,1i}V_{\vep,\sigma,1j})^{1/2}(V_{\vep,\sigma,1i}V_{\vep,\sigma,1j})^{1/2}
\end{equation}
and then use Cauchy-Schwarz to obtain
\begin{equation}
\begin{split}
\left|\ip{\wh{(\tau_2\chi^{(1)})}q_j\Phi^\vep}{p_1p_iV_{\vep,\sigma,1i}V_{\vep,\sigma,1j}p_jp_1\wh{(\tau_2\chi^{(1)})}q_i\Phi^\vep}_{L_{\ux_N}^2}\right| &\leq \|(V_{\vep,\sigma,1i}V_{\vep,\sigma,1j})^{1/2}p_1p_i\wh{(\tau_2\chi^{(1)})}q_j\Phi^\vep\|_{L_{\ux_N}^2} \\
&\hspace{10mm}\times \|(V_{\vep,\sigma,1i}V_{\vep,\sigma,1j})^{1/2}p_jp_1\wh{(\tau_2\chi^{(1)})} q_i\Phi^\vep\|_{L_{\ux_N}^2}.
\end{split}
\end{equation}
From Young's inequality for products and the symmetry of $\Phi^\vep$ under permutation of particle labels, we then find that
\begin{equation}
\label{eq:A1_chi(1)_op}
\eqref{eq:A1_chi(1)_sym} \leq \sum_{2\leq i\neq j\leq N} \ip{\Phi^\vep}{\wh{(\tau_2\chi^{(1)})}q_jp_1p_iV_{\vep,\sigma,1i}V_{\vep,\sigma,1j}p_1p_iq_j\wh{(\tau_2\chi^{(1)})}\Phi^\vep}_{L_{\ux_N}^2}.
\end{equation}
Next, by computation of its integral kernel, we see that
\begin{equation}
p_iV_{\vep,\sigma,1i}V_{\vep,\sigma,1j}p_i = p_i (V_{\vep,\sigma}\ast |\phi|^2)_1 V_{\vep,\sigma,1j},
\end{equation}
and
\begin{equation}
(p_1(V_{\vep,\sigma}\ast|\phi|^2)_1 V_{\vep,\sigma,1j}p_1) = p_1\paren*{V_{\vep,\sigma}\ast (|\phi|^2(V_{\vep,\sigma}\ast|\phi|^2))}_j.
\end{equation}
By Young's inequality with $\|V_{\vep,\sigma}\|_{L^1}=1$, followed by H\"older's inequality, and then another application of Young's, we have that
\begin{equation}
\|\paren*{V_{\vep,\sigma}\ast (|\phi|^2(V_{\vep,\sigma}\ast|\phi|^2))}\|_{L_x^\infty} \leq \|\phi\|_{L_x^\infty}^2\|V_{\vep,\sigma}\ast|\phi|^2\|_{L_x^\infty} \leq \|\phi\|_{L_x^\infty}^4,
\end{equation}
which implies that
\begin{equation}
\label{eq:p1pi_V_opnorm}
\|p_1p_i V_{\vep,\sigma,1i}V_{\vep,\sigma,1j}p_1p_i\|_{L_{\ux_N}^2\rightarrow L_{\ux_N}^2} \leq \|\phi\|_{L_x^\infty}^4.
\end{equation}
Applying this last estimate to the right-hand side of \eqref{eq:A1_chi(1)_op} and the symmetry of $\Phi^\vep$, we obtain that
\begin{equation}
\label{eq:A1_chi(1)_n^2_app}
|A_1| \lesssim \|\phi\|_{L_x^\infty}^4 \sum_{2\leq i\neq j\leq N} \|\wh{(\tau_2\chi^{(1)})}q_j\Phi^\vep\|_{L_{\ux_N}^2}^2 \leq N^2\|\phi\|_{L_x^\infty}^4 \|\wh{(\tau_2\chi^{(1)})}q_1\Phi^{\vep}\|_{L_{\ux_N}^2}^2 \leq N^2\|\phi\|_{L_x^\infty}^4 \|\wh{(\tau_2\chi^{(1)})}\wh{n}\Phi^\vep\|_{L_{\ux_N}^2}^2,
\end{equation}
where the ultimate inequality follows by application of \cref{lem:mN_est}\ref{item:mN_q1} to the factor $\|\wh{(\tau_2\chi^{(1)})}q_1\Phi^{\vep}\|_{L_{\ux_N}^2}$. In order to estimate the last expression, we claim that
\begin{equation}
\label{eq:chi(1)_n^2_bnd}
(\tau_2\chi^{(1)})(k)n(k) \leq N^{-\delta/2}, \qquad \forall k\in\{0,\ldots,N\}.
\end{equation}
Indeed, recalling from \eqref{eq:chi(1)(2)_def} that $\chi^{(1)} = 1_{\leq N^{1-\delta}}$, where $\delta\in (0,1)$, we see that
\begin{equation}
(\tau_2\chi^{(1)})(k)n(k) = 1_{\leq N^{1-\delta}}(k+2)1_{\geq 0}(k)\sqrt{\frac{(k+2)-2}{N}} \leq 1_{\leq N^{1-\delta}}(k)\sqrt{\frac{N^{1-\delta}}{N}-\frac{2}{N}},
\end{equation}
from which the claim follows. Applying this estimate to the right-hand side of \eqref{eq:A1_chi(1)_n^2_app} leads to the conclusion
\begin{equation}
\label{eq:A1_chi(1)_fin}
|A_1| \lesssim N^{2-\delta}\|\phi\|_{L_x^\infty}^4.
\end{equation}
Now using the identity
\begin{equation}
\label{eq:p1_V_opnorm}
p_1V_{\vep,\sigma,1i}p_1V_{\vep,\sigma,1j}p_1 = p_1(V_{\vep,\sigma}\ast|\phi|^2)_i (V_{\vep,\sigma}\ast|\phi|^2)_j,
\end{equation}
which follows from examination of the integral kernel, and arguing similarly as for $A_1$, we find that
\begin{align}
|A_2| \leq \|V_{\vep,\sigma}\ast|\phi|^2\|_{L_x^\infty}^2 \sum_{2\leq i\neq j\leq N} \|q_j\wh{(\tau_2\chi^{(1)})}\Phi^\vep\|_{L_{\ux_N}^2} \|q_i\wh{(\tau_2\chi^{(1)})}\Phi^\vep\|_{L_{\ux_N}^2} \lesssim N^{2-\delta}\|\phi\|_{L_x^\infty}^4. \label{eq:A2_chi(1)_fin}
\end{align}
Thus, we conclude from \eqref{eq:A1_chi(1)_fin} and \eqref{eq:A2_chi(1)_fin} that
\begin{equation}
\label{eq:A_chi(1)_fin}
|A| \lesssim N^{2-\delta}\|\phi\|_{L_x^\infty}^4.
\end{equation}

To conclude the estimate for $\mathrm{Term}_{3,2,1}$ defined in \eqref{eq:T_32j_def} above, we insert the estimate \eqref{eq:B_chi1_fin} for $B$ and the estimate \eqref{eq:A_chi(1)_fin} for $A$ into the right-hand side of \eqref{eq:T_321_AB}, obtaining
\begin{equation}
\label{eq:T_321_fin}
\mathrm{Term}_{3,2,1} \lesssim \frac{1}{N-1}\paren*{\sqrt{(N-1)N^{\sigma}\|\phi\|_{L_x^4}^4} + \sqrt{N^{2-\delta}\|\phi\|_{L_x^\infty}^4}} \lesssim \frac{\|\phi\|_{L_x^4}^2}{N^{(1-\sigma)/2}} + \frac{\|\phi\|_{L_x^\infty}^2}{N^{\delta/2}}. 
\end{equation}

It remains for us to estimate $\mathrm{Term}_{3,2,2}$, which we recall from \eqref{eq:T_32j_def} is defined by
\begin{equation}
\mathrm{Term}_{3,2,2} = \left|\ip{\Phi^\vep}{p_1p_2V_{\vep,\sigma,12}\wh{\nu}\wh{\chi^{(2)}}q_1q_2\Phi^\vep}_{L_{\ux_N}^2}\right|.
\end{equation}
Writing $\wh{\nu}=\wh{\nu}^{1/2}\wh{\nu}^{1/2}$ and using the same symmetrization trick as above, we find that
\begin{align}
\mathrm{Term}_{3,2,2} &=\frac{1}{N-1}\left|\ip{\Phi^\vep}{\sum_{i=2}^N p_1p_i V_{\vep,\sigma,1i}q_iq_1\wh{\chi^{(2)}}\wh{\nu}^{1/2}\wh{\nu}^{1/2}\Phi^\vep}_{L_{\ux_N}^2}\right| \nn\\
&\leq \frac{1}{N-1}\|\wh{\nu}^{1/2}q_1\Phi^{\vep}\|_{L_{\ux_N}^2} \sqrt{\sum_{i,j=2}^N \ip{\Phi^\vep}{p_1p_iV_{\vep,\sigma,1i}q_1q_i\wh{\chi^{(2)}}\wh{\nu}q_1q_jV_{\vep,\sigma,1j}p_jp_1\Phi^\vep}_{L_{\ux_N}^2}}, \label{eq:sum_split}
\end{align}
where the ultimate inequality follows by Cauchy-Schwarz and expanding the $L_{\ux_N}^2$ norm of the second factor. By the $\nu$ estimate \eqref{eq:nu_bnd} together with \cref{lem:mN_est}\ref{item:mN_q1},
\begin{equation}
\label{eq:nu_1/2_q1}
\|\wh{\nu}^{1/2}q_1\Phi^\vep\|_{L_{\ux_N}^2} = \sqrt{\ip{\Phi^\vep}{\wh{\nu}q_1\Phi^\vep}_{L_{\ux_N}^2}} \lesssim \sqrt{\ip{\Phi^\vep}{\wh{n}^{-1}q_1\Phi^\vep}_{L_{\ux_N}^2}} \lesssim \sqrt{\beta_\vep}.
\end{equation}
Thus, splitting the sum $\sum_{i,j}=\sum_{i} + \sum_{i\neq j}$ in the second factor of \eqref{eq:sum_split} and applying the embedding $\ell^{1/2}\subset\ell^1$, we obtain that
\begin{equation}
\label{eq:T_322_AB}
\mathrm{Term}_{3,2,2} \leq \frac{\sqrt{\beta_\vep}}{N-1} \paren*{\sqrt{A}+\sqrt{B}},
\end{equation}
where
\begin{align}
B &\coloneqq \sum_{i=2}^N \ip{\Phi^\vep}{p_1p_iV_{\vep,\sigma,1i}q_1q_i\wh{\chi^{(2)}}\wh{\nu}V_{\vep,\sigma,1i}p_ip_1\Phi^\vep}_{L_{\ux_N}^2}, \label{eq:chi(2)_B_def}\\
A &\coloneqq \sum_{2\leq i\neq j\leq N} \ip{\Phi^\vep}{p_1p_iV_{\vep,\sigma,1i}q_1q_i \wh{\chi^{(2)}} \wh{\nu}q_j V_{\vep,\sigma,1j}p_jp_1\Phi^\vep}_{L_{\ux_N}^2}. \label{eq:chi(2)_A_def}
\end{align}
Note that in contrast to the inequality \eqref{eq:T_321_AB} for $\mathrm{Term}_{3,2,1}$, we have a factor of $\sqrt{\beta_\vep}$ in the right-hand side of inequality \eqref{eq:T_322_AB}.

We first dispense with the easy case $B$. We recall from \eqref{eq:chi(1)(2)_def} that $\chi^{(2)}=1_{>N^{1-\delta}}$, which together with the $\nu$ bound \eqref{eq:nu_bnd} implies the estimate
\begin{equation}
\label{eq:chi(2)_nu_bnd}
\chi^{(2)}(k)\nu(k) \lesssim 1_{>N^{1-\delta}}(k)n^{-1}(k) = 1_{>N^{1-\delta}}(k) \sqrt{\frac{N}{k}} < N^{\delta/2}, \qquad \forall k\in\Z.
\end{equation}
Therefore, we have the $L_{\ux_N}^2$ operator norm estimate
\begin{equation}
\|q_1q_i \wh{\chi^{(2)}}\wh{\nu}\|_{L_{\ux_N}^2\rightarrow L_{\ux_N}^2} \lesssim N^{\delta/2}, \qquad \forall i\in\{1,\ldots,N\},
\end{equation}
which implies that
\begin{equation}
B\lesssim N^{\delta/2} \sum_{i=2}^N \|V_{\vep,\sigma,1i}p_1p_i\Phi^\vep\|_{L_{\ux_N}^2}^2 = (N-1)N^{\delta/2}\|V_{\vep,\sigma,12}p_1p_2\Phi^\vep\|_{L_{\ux_N}^2}^2,
\end{equation}
where the ultimate identity follows from the symmetry of $\Phi^\vep$. Since by Cauchy-Schwarz and Young's inequality,
\begin{equation}
p_1p_2V_{\vep,\sigma,12}^2p_1p_2 = \||\phi|^2(V_{\vep,\sigma}^2\ast |\phi|^2)\|_{L_x^1} p_1p_2 \lesssim N^{\sigma} \|\phi\|_{L_x^4}^4p_1p_2,
\end{equation}
where we also use $\|V_{\vep,\sigma}\|_{L^2}^2\lesssim N^\sigma$, we conclude that
\begin{equation}
\label{eq:chi(2)_B_fin}
B\lesssim N^{1+\frac{\delta}{2}+\sigma}\|\phi\|_{L_x^4}^4.
\end{equation}

For the hard case $A$, we again use \cref{lem:shift} as in \eqref{eq:chi(1)_A_split} to write $A=A_1+A_2$, where
\begin{align}
A_1 &\coloneqq \sum_{2\leq i\neq j\leq N} \ip{\Phi^\vep}{p_1p_iq_j\wh{(\tau_2\chi^{(2)})}\wh{(\tau_2\nu)}^{1/2}V_{\vep,\sigma,1i}V_{\vep,\sigma,1j}\wh{(\tau_2\chi^{(2)})}\wh{(\tau_2\nu)}^{1/2}q_ip_jp_1\Phi^\vep}_{L_{\ux_N}^2}, \\
A_2 &\coloneqq -\sum_{2\leq i\neq j\leq N} \ip{\Phi^\vep}{p_1p_iq_j\wh{(\tau_2\chi^{(2)})}\wh{(\tau_2\nu)}^{1/2}V_{\vep,\sigma,1i}p_1 V_{\vep,\sigma,1j}\wh{(\tau_2\chi^{(2)})}^{1/2}\wh{(\tau_2\nu)}^{1/2}q_ip_jp_1\Phi^\vep}_{L_{\ux_N}^2}.
\end{align}
For $A_1$, we use that $V_{\vep,\sigma}\geq 0$ to apply Cauchy-Schwarz and exploit the symmetry of $\Phi_{\vep}$ under exchange of particle labels in order to obtain
\begin{equation}
|A_1| \leq \sum_{2\leq i\neq j\leq N} \left|\ip{\Phi^\vep}{q_j\wh{(\tau_2\chi^{(2)})}\wh{(\tau_2\nu)}^{1/2}p_1p_i V_{\vep,\sigma,1i}V_{\vep,\sigma,1j}p_ip_1\wh{(\tau_2\chi^{(2)})}\wh{(\tau_2\nu)}^{1/2}q_j\Phi^\vep}_{L_{\ux_N}^2}\right|.
\end{equation}
Using the $L_{\ux_N}^2$ operator norm estimate \eqref{eq:p1pi_V_opnorm}, we conclude that
\begin{align}
|A_1| \lesssim \|\phi\|_{L_x^\infty}^4\sum_{2\leq i\neq j\leq N} \underbrace{\|\wh{(\tau_2\chi^{(2)})}\wh{(\tau_2\nu)}^{1/2}q_1\Phi^\vep\|_{L_{\ux_N}^2}^2}_{\leq\ip{\Phi^\vep}{\wh{(\tau_2\nu)}q_1\Phi^\vep}_{L_{\ux_N}^2}} \lesssim N^2\|\phi\|_{L_x^\infty}^4 \ip{\Phi^\vep}{\wh{n}\Phi^\vep}_{L_{\ux_N}^2} =N^2\|\phi\|_{L_x^\infty}^4\beta_\vep, \label{eq:chi(2)_A1_fin}
\end{align}
where the penultimate inequality follows from the $\nu$ estimate \eqref{eq:nu_bnd} together with \cref{lem:mN_est}\ref{item:mN_q1} and the ultimate equality is by definition of $\beta_\vep$ (recall \eqref{eq:beta_ep_def}). Next, using the operator identity \eqref{eq:p1_V_opnorm} and arguing similarly as for $A_2$ in the case of $\chi^{(1)}$, we also obtain the estimate
\begin{equation}
\label{eq:chi(2)_A2_fin}
|A_2| \lesssim N^2\|\phi\|_{L_x^\infty}^4\beta_\vep,
\end{equation}
leading us to conclude that
\begin{equation}
\label{eq:chi(2)_A_fin}
|A| \lesssim N^2\|\phi\|_{L_x^\infty}^4\beta_\vep.
\end{equation}

Inserting the estimates \eqref{eq:chi(2)_B_fin} for $B$ and \eqref{eq:chi(2)_A_fin} for $A$ into the right-hand side of \eqref{eq:T_322_AB}, we find from the normalization $\|\phi\|_{L_x^2}=1$ and Young's inequality for products that
\begin{align}
\mathrm{Term}_{3,2,2} \lesssim \frac{\sqrt{\beta_\vep}}{N-1}\paren*{N\|\phi\|_{L_x^\infty}^2\sqrt{\beta_\vep} + N^{\frac{1+\sigma}{2}+\frac{\delta}{4}}\|\phi\|_{L_x^4}^2} &\lesssim \|\phi\|_{L_x^\infty}^2\beta_\vep + N^{\frac{2(\sigma-1)+\delta}{2}}. \label{eq:T_322_fin}
\end{align}

Collecting the estimates \eqref{eq:T_321_fin} for $\mathrm{Term}_{3,2,1}$ and \eqref{eq:T_322_fin} for $\mathrm{Term}_{3,2,2}$, we find that
\begin{equation}
\label{eq:T_32_fin}
|\mathrm{Term}_{3,2}| \lesssim N^{\frac{\sigma-1}{2}}\|\phi\|_{L_x^4}^2 + N^{-\frac{\delta}{2}}\|\phi\|_{L_x^\infty}^2 + \|\phi\|_{L_x^\infty}^2\beta_\vep + N^{\frac{2(\sigma-1)+\delta}{2}}.
\end{equation}

Now inserting the estimates \eqref{eq:beta_ep_T_31_fin} for $\mathrm{Term}_{3,1}$ and \eqref{eq:T_32_fin} for $\mathrm{Term}_{3,2}$ into the right-hand side of \eqref{eq:T_3_12_def}, we conclude that
\begin{align}
|\mathrm{Term}_3| &\lesssim  N^{-\sigma} + \|\phi\|_{C_x^{1/2}}^2\|\phi\|_{H_x^1}^2\beta_\vep + \|\phi\|_{C_x^{1/2}}^2\|\nabla_1q_1\Phi^\vep\|_{L_{\ux_N}^2}^2 + N^{\frac{\sigma-1}{2}}\|\phi\|_{L_x^4}^2 + N^{-\frac{\delta}{2}}\|\phi\|_{L_x^\infty}^2 +  N^{\frac{2(\sigma-1)+\delta}{2}}, \label{eq:beta_ep_T3_fin}
\end{align}
where we implicitly use the $\|\phi\|_{H_x^1}^2\geq 1$ by the unit mass normalization.
\end{description}

We are now prepared to conclude the proof of the proposition. After a bookkeeping of the estimates \eqref{eq:beta_ep_T_1_fin} for $\mathrm{Term}_1$, \eqref{eq:beta_ep_T_2_fin} for $\mathrm{Term}_2$, and \eqref{eq:beta_ep_T3_fin} for $\mathrm{Term}_3$, we find that
\begin{equation}
\label{eq:RHS_aux}
\begin{split}
\dot{\beta}_\vep &\lesssim \frac{\|\phi\|_{L_x^\infty}^2}{N} + \vep^{1/2}\|\phi\|_{C_x^{1/2}}^2 + \|\phi\|_{L_x^\infty}^2\beta_\vep + \|\nabla_1 q_1\Phi^\vep\|_{L_{\ux_N}^2}\sqrt{\beta_\vep} + \frac{1}{N^\sigma} + \|\phi\|_{C_x^{1/2}}^2\|\phi\|_{H_x^1}^2\beta_\vep\\
&\ph  + \|\phi\|_{C_x^{1/2}}^2\|\nabla_1q_1\Phi^\vep\|_{L_{\ux_N}^2}^2 + \frac{\|\phi\|_{L_x^4}^2}{N^{(1-\sigma)/2}} + \frac{\|\phi\|_{L_x^\infty}^2}{N^{\delta/2}} +  N^{\frac{2(\sigma-1)+\delta}{2}}.
\end{split}
\end{equation}
The desired conclusion now follows from Young's inequality for products, $\|\phi\|_{L_x^2}=1$, and some algebra.
\end{proof}

\subsection{Auxiliary control}
\label{ssec:MR_aux}
We now estimate the auxiliary quantity $\|\nabla_1q_1\Phi_N^\vep\|_{L_{\ux_N}^2}$ appearing in the estimate of \cref{prop:beta_ep_ineq} in terms of $\beta_{N,\vep}$, $N$, and $(E_{N,\vep}^{\Phi_N^\vep}-E^\phi)$. Here, $E_{N,\vep}^{\Phi_N^\vep}=\ip{\Phi_N^\vep}{H_{N,\vep}\Phi_N^\vep}$ is the regularized microscopic energy per particle.
\begin{restatable}[Control of $\|\nabla_1 q_1\Phi_N^\vep\|_{L^2}^2$]{prop}{gradcon}
\label{prop:grad_q1}
For $\kappa\in\{\pm 1\}$, we have the estimate
\begin{equation}
\|\nabla_1 q_1(t)\Phi_N^\vep(t)\|_{L^2(\R^N)}^2 \lesssim E_{N,\vep}^{\Phi_N^\vep}-E^\phi + \vep^{1/2}\|\phi(t)\|_{{C}^{1/2}(\R)}^2+ \|\phi(t)\|_{H^2(\R)}\beta_{N,\vep}(t) + \frac{\|\phi(t)\|_{H^2(\R)}}{\sqrt{N}},
\end{equation}
for every $t\in\R$, uniformly in $\vep>0$ and $N\in\N$.
\end{restatable}
\begin{proof}
As before, we drop the subscript $N$, as the number of particles is fixed throughout the proof. We introduce two parameters $\kappa_1\in (0,1)$ and $\kappa_2>0$, the precise values of which we shall specify momentarily. Using the decomposition $\id = p_1p_2 + (\id-p_1p_2)$ and the normalizations $\|\Phi^\vep\|_{L_{\ux_N}^2}=1=\|\phi\|_{L_x^2}$, we arrive at the identity
\begin{equation}
(1-\kappa_1)\|\nabla_1(\id - p_1p_2)\Phi^\vep\|_{L_{\ux_N}^2}^2 = E_{\vep}^{\Phi^\vep} - E^{\phi} + \sum_{i=1}^6\mathrm{Term}_i,
\end{equation}
where
\begin{align}
\mathrm{Term}_1 &\coloneqq - \|\nabla_1p_1p_2\Phi^\vep\|_{L_{\ux_N}^2}^2+\|\nabla\phi\|_{L_x^2}^2, \\
\mathrm{Term}_2 &\coloneqq -\kappa_2\ip{\Phi^\vep}{p_1p_2\Phi^\vep}_{L_{\ux_N}^2} + \kappa_2, \\
\mathrm{Term}_3 &\coloneqq -\frac{\kappa(N-1)}{2N}\ip{\Phi^\vep}{p_1p_2 V_{\vep,12}p_1p_2\Phi^\vep}_{L_{\ux_N}^2} + \frac{\kappa}{2}\|\phi\|_{L_x^4}^4, \\
\mathrm{Term}_4 &\coloneqq -2\Re{\ip{\nabla_1(\id-p_1p_2)\Phi^\vep}{\nabla_1p_1p_2\Phi^\vep}_{L_{\ux_N}^2}}, \\
\mathrm{Term}_5 &\coloneqq -\frac{\kappa(N-1)}{N}\Re{\ip{\Phi^\vep}{(\id-p_1p_2)V_{\vep,12}p_1p_2\Phi^\vep}_{L_{\ux_N}^2}},\\
\mathrm{Term}_6 &\coloneqq -\frac{\kappa(N-1)}{2N} \|V_{\vep,12}^{1/2}(\id-p_1p_2)\Phi^\vep\|_{L_{\ux_N}^2}^2 - \kappa_1 \|\nabla_1(\id-p_1p_2)\Phi^\vep\|_{L_{\ux_N}^2}^2 -\kappa_2\|(\id-p_1p_2)\Phi^\vep\|_{L_{\ux_N}^2}^2.
\end{align}
We keep the term $E_\vep^{\Phi^\vep}-E^{\phi}$. We want to obtain upper bounds for the moduli of $\mathrm{Term}_1,\ldots,\mathrm{Term}_5$, and we want to show that $\mathrm{Term}_6 \leq 0$ provided that we appropriately choose $\kappa_1,\kappa_2$ depending on $\kappa$.

\begin{description}[leftmargin=*]
\item[Estimate for $\mathrm{Term}_1$]
Since $\nabla_1p_1 = (\ket*{\nabla\phi}\bra*{\phi})_1$, it follows from $1=\|\Phi^\vep\|_{L_{\ux_N}^2}$ that
\begin{align}
\mathrm{Term}_1 = \|\nabla\phi\|_{L_x^2}^2\paren*{1-\ip{\Phi^\vep}{p_1p_2\Phi^\vep}_{L_{\ux_N}^2}} = \ip{\Phi^\vep}{(\id-p_1p_2)\Phi^\vep}_{L_{\ux_N}^2}.
\end{align}
Since $\id-p_1p_2 = q_1p_2+q_2p_1+q_1q_2$, it follows from \cref{rem:alpha_N} and the triangle inequality that
\begin{equation}
\label{eq:3alpha_bnd}
\ip{\Phi^\vep}{(\id-p_1p_2)\Phi^\vep}_{L_{\ux_N}^2} \leq 3\alpha_\vep \lesssim\beta_\vep,
\end{equation}
leading us to conclude that
\begin{equation}
\mathrm{Term}_1 \lesssim \|\nabla\phi\|_{L_x^2}^2\beta_\vep. \label{eq:con_T_1_fin}
\end{equation}
\item[Estimate for $\mathrm{Term}_2$]
Using the identity $\kappa_2\|\Phi^\vep\|_{L_{\ux_N}^2}^2=\kappa_2$ and the estimate \eqref{eq:3alpha_bnd}, we find that
\begin{equation}
\mathrm{Term}_2 =\kappa_2\ip{\Phi^\vep}{(\id-p_1p_2)\Phi^\vep}_{L_{\ux_N}^2} \lesssim \kappa_2\beta_\vep. \label{eq:con_T_2_fin}
\end{equation}
\item[Estimate for $\mathrm{Term}_3$]
First, observe that
\begin{equation}
p_1p_2 V_{12}p_1p_2 = \|\phi\|_{L_x^4}^4p_1p_2 \quad \text{and} \quad p_1p_2 V_{\vep,12}p_1p_2 = \||\phi|^2 (V_\vep \ast |\phi|^2)\|_{L_x^1} p_1p_2.
\end{equation}
So by the triangle inequality,
\begin{align}
\label{eq:T3_fs}
\left|\mathrm{Term}_3\right| &\leq \frac{1}{2}\left|\ip{\Phi^\vep}{p_1p_2(V_{\vep,12}-V_{12})p_1p_2\Phi^\vep}_{L_{\ux_N}^2}\right| +  \frac{\|\phi\|_{L_x^4}^4}{2}\left|-\frac{(N-1)}{N}\ip{\Phi^\vep}{p_1p_2\Phi^\vep}_{L_{\ux_N}^2} + 1\right|.
\end{align}
Since $\|\Phi^\vep\|_{L_{\ux_N}^2}^2=1$, the second term in the right-hand side equals
\begin{equation}
\frac{\|\phi\|_{L_x^4}^4}{2}\left|\frac{1}{N}\ip{\Phi^\vep}{p_1p_2\Phi^\vep}_{L_{\ux_N}^2} + \ip{\Phi^\vep}{(\id-p_1p_2)\Phi^\vep}_{L_{\ux_N}^2}\right| \lesssim \|\phi\|_{L_x^4}^4\paren*{\frac{1}{N}+\beta_\vep},
\end{equation}
where the ultimate inequality follows from the triangle inequality, $\ip{\Phi^\vep}{p_1p_2\Phi^\vep}\leq \|\Phi^\vep\|_{L_{\ux_N}^2}^2 = 1$, and the estimate \eqref{eq:3alpha_bnd}. Again using that $\|\Phi^\vep\|_{L_{\ux_N}^2}=1$, we see that the first term in the right-hand side of \eqref{eq:T3_fs} is bounded by
\begin{equation}
\frac{1}{2} \| |\phi|^2\paren*{(V_\vep\ast|\phi|^2)-|\phi|^2}\|_{L_x^1} \lesssim \|\phi\|_{C_x^{1/2}}^2 \vep^{1/2},
\end{equation}
which follows from the estimate \eqref{eq:op_norm_conv} and $\|\phi\|_{L_x^2}=1$. Therefore,
\begin{align}
\label{eq:con_T_3_fin}
\mathrm{Term}_3 \lesssim \vep^{1/2} \|\phi\|_{C_x^{1/2}}^2 + \|\phi\|_{L_x^4}^4\paren*{\frac{1}{N}+\beta_\vep}.
\end{align}

\item[Estimate for $\mathrm{Term}_4$]
By using the decomposition $\id-p_1p_2 = q_1p_2+q_2p_1+q_1q_2$, the triangle inequality, and the fact that $\comm{q_2}{\nabla_1}=0=q_2p_2$, we see that
\begin{align}
\left|\mathrm{Term}_4\right| &\lesssim \left|\ip{\nabla_1q_1p_2\Phi^\vep}{\nabla_1p_1p_2\Phi^\vep}_{L_{\ux_N}^2} + \underbrace{\ip{\nabla_1 q_2p_1\Phi^\vep}{\nabla_1 p_1p_2\Phi^\vep}_{L_{\ux_N}^2}}_{=0} + \underbrace{\ip{\nabla_1q_1q_2\Phi^\vep}{\nabla_1p_1p_2\Phi^\vep}_{L_{\ux_N}^2}}_{=0}\right| \nn\\
&=\left|\ip{\wh{n}^{-1/2}q_1\Phi^\vep}{\wh{n}^{1/2}(-\Delta_1) p_1p_2\Phi^\vep}_{L_{\ux_N}^2}\right|,
\end{align}
where the ultimate equality follows from integration by parts and writing $\id=\wh{n}^{-1/2}\wh{n}^{1/2}$. The reader will recall the definitions of $n$ and $\wh{n}$ from \cref{def:mn_N}. By Cauchy-Schwarz and $q_1^2=q_1$,
\begin{align}
\left|\ip{\wh{n}^{-1/2}q_1\Phi^\vep}{\wh{n}^{1/2}(-\Delta_1) p_1p_2\Phi^\vep}_{L_{\ux_N}^2}\right| &\leq \|\wh{n}^{-1/2}q_1\Phi^\vep\|_{L_{\ux_N}^2} \|q_1\wh{n}^{1/2}(-\Delta_1)p_1p_2\Phi^\vep\|_{L_{\ux_N}^2} \nn\\
&\leq \sqrt{\beta_\vep}\|q_1\wh{n}^{1/2}(-\Delta_1)p_1p_2\Phi^\vep\|_{L_{\ux_N}^2},
\end{align}
where the ultimate line follows from applying \cref{lem:mN_est}\ref{item:mN_q1} to the first factor in the right-hand side of the first line. By \cref{lem:shift}, we have the operator identity
\begin{equation}
q_1\wh{n}^{1/2}(-\D_1)p_1 = q_1(-\D_1)\wh{(\tau_1 n)}^{1/2}p_1 = q_1(-\D_1)p_1\wh{(\tau_1n)}^{1/2}.
\end{equation}
So writing $q_1 = \id - p_1$ and using the triangle inequality together with the operator norm estimates
\begin{equation}
\|(-\D_1)p_1\|_{L_{\ux_N}^2\rightarrow L_{\ux_N}^2} \leq \|\D\phi\|_{L_x^2} \quad \text{and} \quad \|p_1(-\D_1)p_1\|_{L_{\ux_N}^2\rightarrow L_{\ux_N}^2} \leq \|\nabla\phi\|_{L_x^2}^2,\footnote{This is the only place in this work where the $H^2$ regularity assumption is strictly needed.}
\end{equation}
we find that
\begin{align}
\|q_1\wh{n}^{1/2}(-\D_1)p_1p_2\Phi^\vep\|_{L_{\ux_N}^2} &\leq \|(-\D_1)p_1\wh{(\tau_1 n)}^{1/2}p_2\Phi^\vep\|_{L_{\ux_N}^2} + \|p_1(-\D_1)p_1\wh{(\tau_1 n)}^{1/2}p_2\Phi^\vep\|_{L_{\ux_N}^2} \nn\\
&\leq \paren*{\|\D\phi\|_{L_x^2}+\|\nabla\phi\|_{L_x^2}^2} \|\wh{(\tau_1n)}^{1/2}\Phi^\vep\|_{L_{\ux_N}^2},
\end{align}
where we eliminate $p_2$ using $\|p_2\|_{L_{\ux_N}^2\rightarrow L_{\ux_N}^2}=1$. Using the embedding $\ell^{1/2}\subset\ell^1$, we see that
\begin{equation}
\label{eq:tau1_n_bnd}
(\tau_1 n)(k) = \sqrt{\frac{k+1}{N}}1_{\geq 0}(k+1) \leq \sqrt{\frac{k}{N}}1_{\geq 0}(k) + \frac{1}{\sqrt{N}} = n(k) + \frac{1}{\sqrt{N}}, \qquad \forall k\in\Z.
\end{equation}
By another application of $\ell^{1/2}\subset \ell^1$ together with $\|\Phi^\vep\|_{L_{\ux_N}^2}=1$,
\begin{equation}
\label{eq:tau1_n_L2}
\|\wh{(\tau_1 n)}^{1/2}\Phi^\vep\|_{L_{\ux_N}^2} \leq \sqrt{\beta_\vep} + N^{-1/4}.
\end{equation}
Using Young's inequality for products and interpolation of $H^s$ spaces with $\|\phi\|_{L_x^2}=1$, we obtain that
\begin{align}
\left|\mathrm{Term}_4\right| &\lesssim \paren*{\|\D\phi\|_{L_x^2} + \|\nabla\phi\|_{L_x^2}^2} \sqrt{\beta_\vep} \paren*{\sqrt{\beta_\vep} + N^{-1/4}} \lesssim \|\phi\|_{H_x^2}\paren*{\beta_\vep + N^{-1/2}}. \label{eq:con_T_4_fin}
\end{align}

\item[Estimate for $\mathrm{Term}_5$]
Using the decomposition $\id-p_1p_2 = p_1q_2+p_2q_1+q_1q_2$ together with the triangle inequality and the symmetry of $\Phi^\vep$ under exchange of particle labels, we have that
\begin{align}
\left|\mathrm{Term}_5\right| &\lesssim \left|\ip{\Phi^\vep}{p_1p_2V_{\vep,12}q_1p_2\Phi^\vep}_{L_{\ux_N}^2} + \ip{\Phi^\vep}{p_1p_2V_{\vep,12}q_2p_1\Phi^\vep}_{L_{\ux_N}^2} + \ip{\Phi^\vep}{p_1p_2 V_{\vep,12}q_1q_2\Phi^\vep}_{L_{\ux_N}^2}\right| \nn\\
&\lesssim \underbrace{\left|\ip{\Phi^\vep}{p_1p_2V_{\vep,12}q_1p_2\Phi^\vep}_{L_{\ux_N}^2}\right|}_{\eqqcolon \mathrm{Term}_{5,1}} + \underbrace{\left|\ip{\Phi^\vep}{p_1p_2V_{\vep,12}q_1q_2\Phi^\vep}_{L_{\ux_N}^2}\right|}_{\eqqcolon \mathrm{Term}_{5,2}},
\end{align}

For $\mathrm{Term}_{5,1}$, we note from an examination of its integral kernel that
\begin{equation}
p_1p_2 V_{\vep,12}q_1p_2 = p_1p_2 V_{\vep,1}^\phi q_1,
\end{equation}
where we use the notation $V_{\vep,1}^\phi$ introduced in \eqref{eq:V_(ep,j)_def}. Now writing $\id = \wh{n}^{-1/2}\wh{n}^{1/2}$, we find that
\begin{align}
\mathrm{Term}_{5,1} &= \left|\ip{\Phi^\vep}{p_1p_2 \wh{(\tau_1 n)}^{1/2} V_{\vep,1}^\phi \wh{n}^{-1/2}q_1\Phi^\vep}_{L_{\ux_N}^2}\right| \nn\\
&\leq \|p_1p_2\wh{(\tau_1 n)}^{1/2}\Phi^\vep\|_{L_{\ux_N}^2} \|V_{\vep,1}^\phi \wh{n}^{-1/2} q_1\Phi^\vep\|_{L_{\ux_N}^2}, \label{eq:T_51_app}
\end{align}
where the penultimate line follows from an application of \cref{lem:shift} and the ultimate line follows from Cauchy-Schwarz. Applying the operator norm identity $\|p_j\|_{L^2\rightarrow L^2}= 1$ together with the estimate \eqref{eq:tau1_n_L2} to the first factor in \eqref{eq:T_51_app}, we obtain that
\begin{equation}
\mathrm{Term}_{5,1} \lesssim \paren*{\sqrt{\beta_\vep}+N^{-1/4}}\|V_{\vep,1}^\phi \wh{n}^{-1/2} q_1\Phi^\vep\|_{L_{\ux_N}^2}.
\end{equation}
Now since  $\|V_{\vep,1}^\phi\|_{L_{\ux_N}^2\rightarrow L_{\ux_N}^2} \leq \|\phi\|_{L_x^\infty}^2$, we find that
\begin{equation}
\|V_{\vep,1}^\phi \wh{n}^{-1/2}q_1\Phi^\vep\|_{L_{\ux_N}^2} \leq \|\phi\|_{L_x^\infty}^2 \|\wh{n}^{-1/2} q_1\Phi^\vep\|_{L_{\ux_N}^2} \leq \|\phi\|_{L_x^\infty}^2 \sqrt{\beta_\vep}.
\end{equation}
where the ultimate equality follows from \cref{lem:mN_est}\ref{item:mN_q1} and the trivial fact that $\wh{n}^2 = \wh{m}$. Using the embedding $\ell^{1/2}\subset\ell^1$, we conclude that
\begin{equation}
\label{eq:con_T_51_fin}
\mathrm{Term}_{5,1} \lesssim \|\phi\|_{L_x^\infty}^2\sqrt{\beta_\vep}\paren*{\sqrt{\beta_\vep} + N^{-1/4}} \lesssim \|\phi\|_{L_x^\infty}^2\paren*{\beta_\vep+N^{-1/2}}.
\end{equation}

For $\mathrm{Term}_{5,2}$, we use, as in the proof of \cref{prop:beta_ep_ineq}, the distributional identity \eqref{eq:intro_sgn} to write $V_{\vep,12}=(\nabla_1 X_{\vep,12})$, where $X_{\vep,12} \coloneqq \frac{1}{2}(V_\vep\ast \sgn)(X_1-X_2)$. Using \cref{lem:shift}, we find that
\begin{align}
\mathrm{Term}_{5,2} &= \left|\ip{\Phi^\vep}{p_1p_2 (\nabla_1 X_{\vep,12})q_1 q_2\Phi^\vep}_{L_{\ux_N}^2}\right| \nn\\
&=\left|\ip{\Phi^\vep}{p_1p_2 (\nabla_1 X_{\vep,12}) \wh{n}\wh{n}^{-1} q_1q_2\Phi^\vep}_{L_{\ux_N}^2}\right| \nn\\
&=\left|\ip{\wh{(\tau_2 n)} p_1p_2\Phi^\vep}{(\nabla_1 X_{\vep,12}) \wh{n}^{-1} q_1q_2\Phi^\vep}_{L_{\ux_N}^2}\right|.
\end{align}
Now integrating by parts and then applying the product rule and triangle inequality, we obtain that
\begin{align}
\left|\ip{\wh{(\tau_2 n)} p_1p_2\Phi^\vep}{(\nabla_1 X_{\vep,12}) \wh{n}^{-1} q_1q_2\Phi^\vep}_{L_{\ux_N}^2}\right| &\leq \left|\ip{\nabla_1 \wh{(\tau_2 n)} p_1p_2\Phi^\vep}{X_{\vep,12}\wh{n}^{-1}q_1q_2\Phi^\vep}_{L_{\ux_N}^2}\right| \nn\\
&\ph +\left|\ip{\wh{(\tau_2 n)}p_1p_2\Phi^\vep}{X_{\vep,12} \nabla_1 \wh{n}^{-1}q_1q_2\Phi^\vep}_{L_{\ux_N}^2}\right| \nn\\
&\eqqcolon \mathrm{Term}_{5,2,1} + \mathrm{Term}_{5,2,2}.
\end{align}

We first dispense with the easy case $\mathrm{Term}_{5,2,1}$. By Cauchy-Schwarz and using the operator norm estimates
\begin{equation}
\label{eq:grad_p1_X_ep}
\|\nabla_1 p_1\|_{L_{\ux_N}^2\rightarrow L_{\ux_N}^2}\leq \|\nabla\phi\|_{L_x^2} \quad \text{and} \quad \|X_{\vep,12}\|_{L_{\ux_N}^2\rightarrow L_{\ux_N}^2} \leq \frac{1}{2},
\end{equation}
we obtain that
\begin{equation}
\mathrm{Term}_{5,2,1} \leq \|\nabla\phi\|_{L_x^2} \|\wh{(\tau_2 n)}\Phi^\vep\|_{L_{\ux_N}^2} \|\wh{n}^{-1}q_1q_2\Phi^\vep\|_{L_{\ux_N}^2}.
\end{equation}
By arguing similarly as for the estimates \eqref{eq:tau1_n_bnd} and \eqref{eq:tau1_n_L2}, we find that
\begin{equation}
\label{eq:tau2_n_L2}
\|\wh{(\tau_2 n)}\Phi^\vep\|_{L_{\ux_N}^2} \lesssim \sqrt{\beta_\vep} + \frac{1}{\sqrt{N}},
\end{equation}
and by applying \cref{lem:mN_est}\ref{item:mN_q1q2}, we have that
\begin{equation}
\|\wh{n}^{-1} q_1 q_2\Phi^\vep\|_{L_{\ux_N}^2} \lesssim \sqrt{\beta_\vep}.
\end{equation}
Thus, we conclude that
\begin{equation}
\label{eq:con_T_521_fin}
\mathrm{Term}_{5,2,1} \lesssim \|\nabla\phi\|_{L_x^2}\paren*{\beta_\vep + \frac{1}{N}}.
\end{equation}

For the hard case $\mathrm{Term}_{5,2,2}$, we first use Cauchy-Schwarz and \eqref{eq:grad_p1_X_ep} to obtain
\begin{align}
\mathrm{Term}_{5,2,2} &\leq  \|\wh{(\tau_2 n)}p_1p_2\Phi^\vep\|_{L_{\ux_N}^2}\|\nabla_1 \wh{n}^{-1}q_1q_2\Phi^\vep\|_{L_{\ux_N}^2} \nn\\
&\lesssim \paren*{\sqrt{\beta_\vep}+N^{-1/2}}\|\nabla_1 \wh{n}^{-1}q_1q_2\Phi^\vep\|_{L_{\ux_N}^2},
\end{align}
where the second line follows from applying the estimate \eqref{eq:tau2_n_L2} to the first factor in the right-hand side of the first line. For the remaining factor $\|\nabla_1 \wh{n}^{-1}q_1q_2\Phi^\vep\|_{L_{\ux_N}^2}$, we write $\id=p_1+q_1$ and use the triangle inequality to obtain
\begin{equation}
\|\nabla_1 \wh{n}^{-1}q_1q_2\Phi^\vep\|_{L_{\ux_N}^2} \leq \|p_1\nabla_1\wh{n}^{-1}q_1q_2\Phi^\vep\|_{L_{\ux_N}^2} + \|q_1\nabla_1\wh{n}^{-1}q_1q_2\Phi^\vep\|_{L_{\ux_N}^2}.
\end{equation}
Since $\|p_1\nabla_1\|_{L_{\ux_N}^2\rightarrow L_{\ux_N}^2}\leq \|\nabla\phi\|_{L_x^2}$, it follows that
\begin{equation}
\|p_1 \nabla_1\wh{n}^{-1}q_1q_2\Phi^\vep\|_{L_{\ux_N}^2} \leq \|\nabla\phi\|_{L_x^2}\|\wh{n}^{-1}q_1q_2\Phi^\vep\|_{L_{\ux_N}^2} \lesssim \|\nabla\phi\|_{L_x^2}\sqrt{\beta_\vep},
\end{equation}
where the ultimate inequality follows from applying \cref{lem:mN_est}\ref{item:mN_q1q2} and $\wh{n}^2=\wh{m}$. Next, observe that by \cref{lem:shift}, $q_1\nabla_1\wh{n}^{-1}q_1=q_1\wh{n}^{-1}\nabla_1q_1$, which implies that
\begin{equation}
\|q_1\nabla_1 \wh{n}^{-1}q_1q_2\Phi^\vep\|_{L_{\ux_N}^2} \leq \|\wh{n}^{-1}\nabla_1 q_1q_2\Phi^\vep\|_{L_{\ux_N}^2} = \sqrt{\ip{\nabla_1 q_1\Phi^\vep}{q_2\wh{n}^{-2}\nabla_1 q_1\Phi^\vep}_{L_{\ux_N}^2}},
\end{equation}
where the ultimate equality follows from the fact that $q_2$ commutes with $\wh{n}^{-2}\nabla_1 q_1$ and $q_2^2=q_2$. By the symmetry of $\Phi^\vep$ with respect to permutation of particle labels and the operator identity
\begin{equation}
\frac{1}{N-1}\sum_{i=2}^N q_i\wh{n}^{-2} \leq \paren*{\frac{N}{N-1}}\wh{m}\wh{n}^{-2}\lesssim \id,
\end{equation}
which follows from \cref{rem:alpha_N}, we see that
\begin{align}
\ip{\nabla_1 q_1\Phi^\vep}{q_2\wh{n}^{-2}\nabla_1 q_1\Phi^\vep}_{L_{\ux_N}^2} &=\frac{1}{N-1}\sum_{i=2}^N \ip{\nabla_1 q_1\Phi^\vep}{q_i \wh{n}^{-2}\nabla_1q_1\Phi^\vep}_{L_{\ux_N}^2} \lesssim \|\nabla_1 q_1\Phi^\vep\|_{L_{\ux_N}^2}^2.
\end{align}
Hence,
\begin{equation}
\|q_1\nabla_1 \wh{n}^{-1}q_1q_2\Phi^\vep\|_{L_{\ux_N}^2} \lesssim \|\nabla_1 q_1\Phi^\vep\|_{L_{\ux_N}^2}.
\end{equation}
We therefore conclude from another application of Young's inequality that
\begin{equation}
\label{eq:con_T_522_fin}
\mathrm{Term}_{5,2,2} \lesssim \|\nabla\phi\|_{L_x^2}\paren*{\beta_\vep+N^{-1}} + \paren*{\sqrt{\beta_\vep}+N^{-1/2}}\|\nabla_1 q_1\Phi^\vep\|_{L_{\ux_N}^2}.
\end{equation}

Collecting the estimate \eqref{eq:con_T_521_fin} for $\mathrm{Term}_{5,2,1}$ and the estimate \eqref{eq:con_T_522_fin} for $\mathrm{Term}_{5,2,2}$, we find that
\begin{equation}
\label{eq:con_T_52_fin}
\mathrm{Term}_{5,2} \lesssim \|\nabla\phi\|_{L_x^2}\paren*{\beta_\vep + N^{-1}} + \paren*{\sqrt{\beta_\vep}+N^{-1/2}} \|\nabla_1 q_1\Phi^\vep\|_{L_{\ux_N}^2}.
\end{equation}
Together with the estimate \eqref{eq:con_T_51_fin} for $\mathrm{Term}_{5,1}$, we conclude that
\begin{equation}
\label{eq:con_T_5_fin}
\begin{split}
\left|\mathrm{Term}_5\right| &\lesssim \|\phi\|_{L_x^\infty}^2\paren*{\beta_\vep+N^{-1/2}} + \|\nabla\phi\|_{L_x^2}\paren*{\beta_\vep + N^{-1}} + \paren*{\sqrt{\beta_\vep}+N^{-1/2}}\|\nabla_1 q_1\Phi^\vep\|_{L_{\ux_N}^2}.
\end{split}
\end{equation}

\item[Estimate for $\mathrm{Term}_6$]
We want to show that $\mathrm{Term}_6\leq 0$. We assume here that $\kappa=-1$; otherwise, it is trivial that $\mathrm{Term}_6\leq 0$ and we can take $\kappa_2=0$. Integrating by parts and using Cauchy-Schwarz,
\begin{equation}
\begin{split}
\|V_{\vep,12}^{1/2}(\id-p_1p_2)\Phi^\vep\|_{L_{\ux_N}^2}^2 &\leq \|\nabla_1(\id-p_1p_2)\Phi^\vep\|_{L_{\ux_N}^2} \|(\id-p_1p_2)\Phi^\vep\|_{L_{\ux_N}^2},
\end{split}
\end{equation}
and by Young's inequality for products,
\begin{equation}
\frac{(N-1)}{2N} \|\nabla_1(\id-p_1p_2)\Phi^\vep\|_{L_{\ux_N}^2} \|(\id-p_1p_2)\Phi^\vep\|_{L_{\ux_N}^2} \leq \kappa_1 \|\nabla_1(\id-p_1p_2)\Phi^\vep\|_{L_{\ux_N}^2}^2 + \frac{(N-1)^2}{4N^2\kappa_1}\|(\id-p_1p_2)\Phi^\vep\|_{L_{\ux_N}^2}^2.
\end{equation}
We choose $\kappa_2 > 1/(2\kappa_1)$. Then,
\begin{align}
\mathrm{Term}_6 &= \frac{(N-1)}{2N}\|V_{\vep,12}^{1/2}(\id-p_1p_2)\Phi^\vep\|_{L_{\ux_N}^2}^2 - \kappa_1 \|\nabla_1(\id-p_1p_2)\Phi^\vep\|_{L_{\ux_N}^2}^2 -\kappa_2\|(\id-p_1p_2)\Phi^\vep\|_{L_{\ux_N}^2}^2 \nn\\
&\leq \paren*{\frac{(N-1)^2}{4N^2\kappa_1}-\kappa_2}\|(\id-p_1p_2)\Phi^\vep\|_{L_{\ux_N}^2}^2 \nn\\
&\leq0,
\end{align}
as desired.
\end{description}

Having estimated the terms $\mathrm{Term}_1,\ldots,\mathrm{Term}_6$, we can now complete the proof of the proposition. Combining estimate \eqref{eq:con_T_1_fin} for $\mathrm{Term}_1$, \eqref{eq:con_T_2_fin} for $\mathrm{Term}_2$, \eqref{eq:con_T_3_fin} for $\mathrm{Term}_3$, \eqref{eq:con_T_4_fin} for $\mathrm{Term}_4$, and \eqref{eq:con_T_5_fin} for $\mathrm{Term}_5$, we see that there exists an absolute constant $C>0$ such that
\begin{equation}
\label{eq:con_lb_pre}
\begin{split}
(1-\kappa_1)\|\nabla_1(\id-p_1p_2)\Phi^\vep\|_{L_{\ux_N}^2}^2 &\leq \paren*{E_\vep^{\Phi^\vep} - E^{\phi}} + C\paren*{\vep^{1/2}\|\phi\|_{C_x^{1/2}}^2 + \paren*{\sqrt{\beta_\vep}+N^{-1/2}}\|\nabla_1q_1\Phi^\vep\|_{L_{\ux_N}^2}}\\
&\ph + C\paren*{\paren*{\|\phi\|_{L_x^\infty}^2 + \|\phi\|_{H_x^2}}N^{-1/2} + \paren*{\|\nabla\phi\|_{L_x^2}+\|\phi\|_{L_x^4}^4}N^{-1}} \\
&\ph +C\beta_\vep\paren*{\|\nabla\phi\|_{L_x^2}^2 + \kappa_21_{\{-1\}}(\kappa) + \|\phi\|_{H_x^2} + \|\phi\|_{L_x^\infty}^2 + \|\nabla\phi\|_{L_x^2} + \|\phi\|_{L_x^4}^4}.
\end{split}
\end{equation}
Note that by using Sobolev embedding, the interpolation property of $H^s$ norms, and the normalization $\|\phi\|_{L_x^2}=1$, we can simplify the right-hand side of \eqref{eq:con_lb_pre} to
\begin{equation}
\label{eq:con_lb}
\begin{split}
(1-\kappa_1)\|\nabla_1(\id-p_1p_2)\Phi^\vep\|_{L_{\ux_N}^2}^2 &\leq \paren*{E_\vep^{\Phi^\vep} - E^{\phi}} + C\|\phi\|_{H_x^2}\paren*{N^{-1/2} + \beta_\vep}\\
&\ph+ C\paren*{\vep^{1/2}\|\phi\|_{C_x^{1/2}}^2 + \paren*{\sqrt{\beta_\vep}+N^{-1/2}}\|\nabla_1q_1\Phi^\vep\|_{L_{\ux_N}^2}},
\end{split}
\end{equation}
for some larger absolute constant $C>0$. To close the proof of the lemma, we want to obtain a lower bound for the left-hand side of \eqref{eq:con_lb} in terms $\|\nabla_1 q_1\Phi^\vep\|_{L_{\ux_N}^2}^2$. To this end, we note that
\begin{equation}
\id - p_1p_2 = p_1+q_1-p_1p_2 = p_1q_2+q_1,
\end{equation}
so that by the triangle inequality and the fact that $q_2$ commutes with $\nabla_1$,
\begin{equation}
\|\nabla_1 q_1\Phi^\vep\|_{L_{\ux_N}^2} \leq \|\nabla_1(\id-p_1p_2)\Phi^\vep\|_{L_{\ux_N}^2} + \|\nabla_1 p_1q_2\Phi^\vep\|_{L_{\ux_N}^2}.
\end{equation}
Since $\|\nabla_1p_1\|_{L_{\ux_N}^2\rightarrow L_{\ux_N}^2} \leq \|\nabla\phi\|_{L_x^2}$, it follows that
\begin{equation}
\|\nabla_1 p_1 q_2\Phi^\vep\|_{L_{\ux_N}^2} \leq \|\nabla\phi\|_{L_x^2}\|q_2\Phi^\vep\|_{L_{\ux_N}^2} \leq \|\nabla\phi\|_{L_x^2}\sqrt{\beta_\vep},
\end{equation}
where the ultimate inequality follows from \cref{rem:alpha_N} and $\alpha_\vep\leq \beta_\vep$. Therefore,
\begin{align}
\|\nabla_1(\id-p_1p_2)\Phi^\vep\|_{L_{\ux_N}^2}^2 \geq \paren*{\|\nabla_1q_1\Phi^\vep\|_{L_{\ux_N}^2} - \|\nabla\phi\|_{L_x^2}\sqrt{\beta_\vep}}^2 \geq \frac{3\|\nabla_1q_1\Phi^\vep\|_{L_{\ux_N}^2}^2}{4} - 15\|\nabla\phi\|_{L_x^2}^2\beta_\vep,
\end{align}
where the ultimate inequality follows from application of Young's inequality for products. Inserting the preceding lower bound into the inequality \eqref{eq:con_lb} and rearranging, we find that
\begin{equation}
\begin{split}
\frac{3}{4}\|\nabla_1 q_1\Phi\|_{L_{\ux_N}^2}^2 &\leq \frac{E_\vep^{\Phi^\vep}-E^\phi}{1-\kappa_1} +  \frac{C}{1-\kappa_1}\paren*{\vep^{1/2} \|\phi\|_{C_x^{1/2}}^2 + \paren*{\sqrt{\beta_\vep}+N^{-1/2}}\|\nabla_1 q_1\Phi^\vep\|_{L_{\ux_N}^2}}\\
&\ph+\frac{C\|\phi\|_{H_x^2}}{1-\kappa_1}\paren*{N^{-1/2} + \beta_\vep} + 15\|\nabla\phi\|_{L_x^2}^2\beta_\vep.
\end{split}
\end{equation}
By Young's inequality for products,
\begin{equation}
\frac{C}{1-\kappa_1}\|\nabla_1 q_1\Phi^\vep\|_{L_{\ux_N}^2}\paren*{\sqrt{\beta_\vep}+N^{-1/2}} \leq \frac{4{C}^2}{(1-\kappa_1)^2}\paren*{\beta_\vep+\frac{1}{N}} + \frac{1}{4}\|\nabla_1 q_1\Phi^\vep\|_{L_{\ux_N}^2}^2,
\end{equation}
The desired conclusion now follows after some algebra.
\end{proof}

\subsection{Proof of \cref{prop:beta_evol}}
\label{ssec:MR_beta_evol}
We now use the results of the previous subsections to send $\vep\rightarrow 0^+$ and obtain an inequality for $\beta_N$, thereby proving \cref{prop:beta_evol}.

\begin{proof}[Proof of \cref{prop:beta_evol}]
Applying \cref{prop:grad_q1} to factors $\|\nabla_1q_1\Phi_N^\vep\|_{L_{\ux_N}^2}$ appearing in the right-hand side of the inequality given by \cref{prop:beta_ep_ineq} and using the majorization $\|\phi\|_{H_x^1}^2\leq \|\phi\|_{H_x^2}$ together with a bit of algebra, we obtain the point-wise estimate
\begin{equation}
\label{eq:dot_beta_op}
\begin{split}
\dot{\beta}_{N,\vep} &\lesssim \frac{\|\phi\|_{L_x^\infty}^2}{N} + \vep^{1/2}\|\phi\|_{C_x^{1/2}}^2 + \frac{(1+\|\phi\|_{C_x^{1/2}}^2)\|\phi\|_{H_x^2}}{\sqrt{N}} + \frac{1}{N^\sigma} + \frac{\|\phi\|_{L_x^4}^4}{N^{(1-\sigma)/2}} + \frac{\|\phi\|_{L_x^\infty}^2}{N^{\delta/2}} + N^{\frac{2(\sigma-1)+\delta}{2}}\\
&\ph+ \paren*{1+\|\phi\|_{C_x^{1/2}}^2}\|\phi\|_{H_x^2}\beta_{N,\vep} +\paren*{1+\|\phi\|_{C_x^{1/2}}^2}\paren*{E_{N,\vep}^{\Phi_N^\vep}-E^\phi + \vep^{1/2}\|\phi\|_{C_x^{1/2}}^2}.
\end{split}
\end{equation}
We now optimize the choice of $\delta,\sigma\in (0,1)$ by requiring that
\begin{equation}
1-\sigma = \delta \quad \text{and} \quad \sigma = \frac{1-\sigma}{2},
\end{equation}
which, after some algebra, implies that $(\delta,\sigma)=(2/3,1/3)$. Inserting this choice of $(\delta,\sigma)$ into the right-hand side of inequality \eqref{eq:dot_beta_op} and using Sobolev embedding, the interpolation property of the $H^s$ norm, and the higher conservation laws of the NLS, we obtain
\begin{equation}
\begin{split}
\dot{\beta}_{N,\vep} &\lesssim \frac{\|\phi_0\|_{H^2}^2}{\sqrt{N}} + \frac{\|\phi_0\|_{H^1}^2}{N^{1/3}} + \|\phi_0\|_{H^2}^2\beta_{N,\vep} + \|\phi_0\|_{H^1}^2\paren*{E_{N,\vep}^{\Phi_{N}^\vep}-E^\phi + \vep^{1/2}\|\phi_0\|_{H^1}^2}.
\end{split}
\end{equation}

Integrating both sides of the preceding inequality over the interval $[0,t]$ and applying the fundamental theorem of calculus, we obtain that
\begin{equation}
\begin{split}
\beta_{N,\vep}(t) &\leq \beta_{N,\vep}(0) + C\|\phi_0\|_{H^2}^2\int_0^t ds \beta_{N,\vep}(s) \\
&\ph + Ct\paren*{\frac{\|\phi_0\|_{H^2}^2}{\sqrt{N}} + \frac{\|\phi_0\|_{H^1}^2}{N^{1/3}} + \|\phi_0\|_{H^1}^2\paren*{E_{N,\vep}^{\Phi_N^\vep}-E^\phi + \vep^{1/2}\|\phi_0\|_{H^1}^2}},
\end{split}
\end{equation}
where $C>0$ is an absolute constant. So applying the Gronwall-Bellman inequality, specifically \cite[Theorem 1.3.1]{Pachpatte1998}, we find that
\begin{equation}
\label{eq:beta_ep_gron}
\begin{split}
\beta_{N,\vep}(t) &\leq \paren*{\beta_{N,\vep}(0)+ Ct\paren*{\frac{\|\phi_0\|_{H^2}^2}{\sqrt{N}} + \frac{\|\phi_0\|_{H^1}^2}{N^{1/3}} + \|\phi_0\|_{H^1}^2\paren*{E_{N,\vep}^{\Phi_N^\vep}-E^\phi + \vep^{1/2}\|\phi_0\|_{H^1}^2}} } e^{Ct\|\phi_0\|_{H^2}^2}.
\end{split}
\end{equation}
We now send $\vep\rightarrow 0^+$ in both sides of inequality \eqref{eq:beta_ep_gron}. By \cref{lem:beta_ep_conv}, we have that $\beta_{N,\vep}(t)\rightarrow\beta_N(t)$ uniformly on compact intervals of time. Recalling the definition of $E_{N,\vep}^{\Phi_N^\vep}$, we see that
\begin{align}
E_{N,\vep}^{\Phi_N^\vep} &= \|\nabla_1\Phi_{N,0}\|_{L^2(\R^N)}^2 + \frac{\kappa(N-1)}{2N}\ip{\Phi_{N,0}}{V_{\vep,12}\Phi_{N,0}}_{L^2(\R^N)}.
\end{align}
It is straightforward to show that $V_{\vep,12}\Phi_{N,0} \rightarrow V_{12}\Phi_{N,0}$ in $H^{-1}(\R^N)$ as $\vep\rightarrow 0^+$. Therefore,
\begin{align}
\lim_{\vep\rightarrow 0^+} E_{N,\vep}^{\Phi_N^\vep}  &= E_N^{\Phi},
\end{align}
which completes the proof of the proposition.
\end{proof}

\section{Proofs of \Cref{thm:mainH} and \Cref{thm:mainL}}
\label{sec:proof_main}
In this last section, we show how \cref{prop:beta_evol} implies \cref{thm:mainH} and \cref{thm:mainL}. We first recall two technical lemmas from \cite{KP2010}.

\begin{lemma}[{\cite[Lemma 2.1]{KP2010}}]
\label{lem:EkE1}
Let $k\in\N$, and let $\{\gamma^{(j)}\}_{j=1}^k$ be a sequence of nonnegative, trace-class operators on $L_{sym}^2(\R^j)$, for $j\in\{1,\ldots,k\}$, with unit trace and such that
\begin{equation}
\Tr_{j+1}\gamma^{(j+1)} = \gamma^{(j)}, \qquad \forall j\in\{1,\ldots,k-1\}.
\end{equation}
Let $\varphi\in L^2(\R)$ satisfy $\|\varphi\|_{L^2}=1$. Then
\begin{equation}
1-\ip{\varphi^{\otimes k}}{\gamma^{(k)}\varphi^{\otimes k}} \leq k\paren*{1-\ip{\varphi}{\gamma^{(1)}\varphi}}.
\end{equation}
\end{lemma}

\begin{lemma}[{\cite[Lemma 2.3]{KP2010}}]
\label{lem:EkRk}
Let $k\in\N$, and let $\gamma^{(k)}$ be a nonnegative self-adjoint trace-class operator on $L_{sym}^2(\R^k)$ with unit trace (i.e. a density matrix). Let $\varphi\in L^2(\R)$ with $\|\varphi\|_{L^2}=1$. Then
\begin{align}
1-\ip{\varphi^{\otimes k}}{\gamma^{(k)}\varphi^{\otimes k}} &\leq \Tr_{1,\ldots,k}\left|\gamma^{(k)}-\ket*{\varphi^{\otimes k}}\bra*{\varphi^{\otimes k}}\right| \leq \sqrt{8\paren*{1-\ip{\varphi^{\otimes k}}{\gamma^{(k)}\varphi^{\otimes k}}}}.
\end{align}
\end{lemma}

\cref{thm:mainH} is now an immediate consequence of \Cref{lem:EkE1,lem:EkRk} together with \cref{prop:beta_evol} and a little bookkeeping. Therefore, we omit the details. Now following the outline discussed in \cref{ssec:intro_RM}, we upgrade \cref{thm:mainH} to an estimate whose right-hand side only requires the $L^2(\R)$ norm of $\phi_0$, thereby proving \cref{thm:mainL}

\begin{proof}[Proof of \cref{thm:mainL}]
Fix $\eta >0$, and suppose that $\Phi_{N,0} = (P_{\leq (\log N)^\eta}\phi_0)^{\otimes N}/\|P_{\leq (\log N)^\eta}\phi_0\|_{L^2(\R)}^N$. Let $\rho: [1,\infty] \rightarrow [0,\infty)$ be a rate function such that $\rho(r)\rightarrow \infty$ as $r\rightarrow\infty$. We will choose $\rho$ momentarily. Let $\phi_{N}$ be the NLS solution with initial datum $\phi_{N,0}$ as in \eqref{eq:nls_moll}. By the aforementioned boundedness of $H^k$ norms for the 1D cubic NLS and Plancherel's theorem,
\begin{equation}
\label{eq:phi_N_H2}
\|\phi_{N}\|_{L_t^\infty H_x^s(\R\times \R)} \lesssim \|\phi_{N,0}\|_{H^s(\R)} \lesssim_s \rho(N)^s\frac{\|P_{\leq \rho(N)}\phi_0\|_{L^2(\R)}}{\|P_{\leq \rho(N)}\phi_0\|_{L^2(\R)}} = \rho(N)^s.
\end{equation}
So applying \cref{thm:mainH} with $\phi$ replaced by $\phi_N$, we find that
\begin{equation}
\label{eq:H2_pre_rho}
\begin{split}
&\sup_{-t\leq s\leq t} \Tr_{1,\ldots,k}\left|\gamma_N^{(k)}(s)-\ket*{\phi_N^{\otimes k}}\bra*{\phi_N^{\otimes k}}(s)\right| \\
&\lesssim_k \paren*{\beta_{N}(\Phi_{N,0},\phi_{N,0})^{1/2} + t^{1/2}\rho(N)\paren*{\frac{1}{N^{1/6}} + \frac{\rho(N)}{N^{1/4}} + |E_N^{\Phi}-E^{\phi_N}|^{1/2}}}e^{C\rho(N)^4 t} .
\end{split}
\end{equation}
Now,
\begin{equation}
\begin{split}
E_N^{\Phi_N} - E^{\phi_N} &=  \frac{\|\nabla P_{\leq (\log N)^\eta}\phi_0\|_{L^2(\R)}^2}{\|P_{\leq (\log N)^\eta}\phi_0\|_{L^2(\R)}^2} - \frac{\|\nabla P_{\leq \rho(N)}\phi_0\|_{L^2(\R)}^2}{\|P_{\leq \rho(N)}\phi_0\|_{L^2(\R)}^2}\\
&\ph + \frac{\kappa}{2}\paren*{\frac{(N-1)\|P_{\leq (\log N)^\eta}\phi_0\|_{L^4(\R)}^4}{N\|P_{\leq (\log N)^\eta}\phi_0\|_{L^2(\R)}^4} - \frac{\|P_{\leq \rho(N)}\phi_0\|_{L^4(\R)}^4}{\|P_{\leq \rho(N)}\phi_0\|_{L^2(\R)}^4}}.
\end{split}
\end{equation}
To cancel the kinetic energy, we choose $\rho(N) = (\log N)^\eta$, which also implies that $\beta_N(\Phi_{N,0},\phi_{N,0})=0$. Moreover, Bernstein's lemma implies
\begin{equation}
|E_N^{\Phi_N}-E^{\phi_N}| \lesssim \frac{(\log N)^\eta}{N},
\end{equation}
provided $N$ is sufficiently large. After some algebra, we see that the right-hand side of inequality \eqref{eq:H2_pre_rho} is $\lesssim$
\begin{equation}
\label{eq:H2_post_rho}
\begin{split}
t^{1/2}\paren*{\frac{(\log N)^{\eta}}{N^{1/6}} + \frac{(\log N)^{2\eta}}{N^{1/4}} + \frac{(\log N)^{3\eta/2}}{N^{1/2}}}e^{Ct(\log N)^{4\eta}} \lesssim \frac{t^{1/2}(\log N)^\eta e^{Ct(\log N)^{4\eta}}}{N^{1/6}}.
\end{split}
\end{equation}
If $0<\eta<1/4$, then $(\ln x)^{4\eta} \ll \ln x$ as $x\rightarrow\infty$. Since $\ln x \ll x$ as $x\rightarrow\infty$, it follows that the expression \eqref{eq:H2_post_rho} tends to zero as $N\rightarrow\infty$, locally uniformly in $t$.

Next, we use \cref{prop:nls_gwp} to obtain that
\begin{align}
\|\phi_N-\phi\|_{L_s^\infty L_x^2([-t,t]\times\R)} &\lesssim \|\phi_{N,0}-\phi_0\|_{L^2(\R)} e^{Ct^{1/2}(\|\phi_N\|_{S^0([-t,t]\times\R)}^2 + \|\phi\|_{S^0([-t,t]\times\R)}^2)},
\end{align}
for some absolute constant $C>0$. Applying the growth bound \eqref{eq:Strich_bbd} to each of the Strichartz norms in the exponent, we find that
\begin{equation}
\label{eq:phiN_dep}
\|\phi_N-\phi\|_{L_s^\infty L_x^2([-t,t]\times\R)} \lesssim \|\phi_{N,0}-\phi_0\|_{L^2(\R)} e^{C't^{5/2}} \lesssim \|P_{>(\log N)^{\eta}}\phi_0\|_{L^2(\R)}e^{C't^{5/2}},
\end{equation}
where $C'\geq C$. By using H\"older's inequality and the estimate \eqref{eq:phiN_dep}, we obtain that
\begin{equation}
\label{eq:phiN_err}
\sup_{-t\leq s\leq t} \Tr\left|\ket*{\phi_N}\bra*{\phi_N}(s)-\ket*{\phi}\bra*{\phi}(s)\right| \lesssim \|P_{>(\log N)^{\eta}}\phi_0\|_{L^2(\R)} e^{C't^{5/2}},
\end{equation}
which evidently tends to zero as $N\rightarrow\infty$. \cref{lem:EkE1} and \cref{lem:EkRk} then yield an estimate for the analogous $k$-particle density matrices.

Finally, we conclude the proof of \cref{thm:mainL} by applying the triangle inequality as in \eqref{eq:TI_intro} and then using the bounds \eqref{eq:H2_post_rho} and \eqref{eq:phiN_err} for the first and second terms on the right-hand side, respectively.
\end{proof}

\bibliographystyle{siam}
\bibliography{GPHam}
\end{document}